\documentclass[sigconf]{acmart}
\usepackage{natbib}
\usepackage{graphicx}
\usepackage{amsthm}
\newtheorem{theorem}{Theorem}[section]
\newtheorem{corollary}{Corollary}[section]
\newtheorem{lemma}{Lemma}[section]
\newtheorem{claim}{Claim}[section]
\makeatletter
\newcommand{\vo}{\vec{o}\@ifnextchar{^}{\,}{}}
\newcommand{\norm}[1]{\left\lVert#1\right\rVert}
\makeatother
\renewcommand{\footnotesize}{\fontsize{8pt}{6pt}\selectfont}

\usepackage{subcaption}
\usepackage{wrapfig}
\usepackage[ruled]{algorithm2e}
\usepackage{booktabs}
\usepackage{pifont}
\usepackage[flushleft]{threeparttable}
\begin{document}

\fancyhead{}
\copyrightyear{2021} 
\acmYear{2021} 
\acmConference[CIKM '21]{Proceedings of the 30th ACM International Conference on Information and Knowledge Management}{November 1--5, 2021}{Virtual Event, Australia}
\acmBooktitle{Proceedings of the 30th ACM International Conference on Information and Knowledge Management (CIKM '21), Nobember 1--5, 2021, Virtual Event, Australia}\acmDOI{10.1145/XXXXXXXXXX}
\acmISBN{XXXXXXX-XX/20/10}
\settopmatter{printacmref=true}

%\title{HopfE:  Representational Learning of Knowledge Graphs using Inverse Hopf Fibrations}
\title{HopfE: Knowledge Graph Representation Learning using Inverse Hopf Fibrations}

\author{Anson Bastos}
\email{cs20resch11002@iith.ac.in}
\affiliation{%
  \institution{IIT, Hyderabad}
  %\city{Aachen}
  \country{India}
}

\author{Kuldeep Singh}
\email{kuldeep.singh1@cerence.com}
\affiliation{%
  \institution{Cerence GmbH and Zerotha Research}
%   \& Zerotha Research,
  %\city{Aachen}
  \country{Germany}
}

\author{Abhishek Nadgeri}
\email{abhishek.nadgeri@rwth-aachen.de}
\affiliation{%
  \institution{RWTH Aachen and Zerotha Research}
  %\city{Aachen}
  \country{Germany}
}

\author{Saeedeh Shekarpour}
\email{sshekarpour1@udayton.edu}
\affiliation{%
  \institution{University of Dayton}
  %\city{Hannover}
   \country{USA}
}

\author{Isaiah Onando Mulang'}
\email{mulang.onando@ibm.com}
\affiliation{%
  \institution{IBM Research and Zerotha Research}
 %\city{Sankt Augustin}
  \country{Kenya}
}

\author{Johannes Hoffart}
\email{johannes.hoffart@gs.com}
\affiliation{%
   \institution{Goldman Sachs}
  \country{Germany}
}
 
%\title{HopfE: Semantically Enriched Knowledge Graph Embeddings using Inverse Hopf Fibrations}
%\title{HopfE: Knowledge Graph Embeddings using Inverse Hopf Fibrations}
%\title{HopfE:  Interpretable and Expressive Knowledge Graph Embeddings using Inverse Hopf Fibrations}
%\title{HopfE: Towards Geometric Interpretability and Expressiveness of Knowledge Graph Embeddings using Inverse Hopf Fibrations}

\renewcommand{\shortauthors}{Bastos et al.}

%%%%%%%%%%%%%% custom commands
\newcommand{\NAME}{\textsc{Plumber}}
%%%%%%%%%%%%%%

\begin{abstract}

%Recently, several Knowledge Graph Embedding (KGE) approaches have been devised to model entitiesand relations in vector space for missing link prediction in a KG. 
Recently, several Knowledge Graph Embedding (KGE) approaches have been devised to represent entities
and relations in a dense vector space and employed in downstream tasks such as link prediction.
A few KGE techniques address \textit{interpretability}, i.e., mapping the connectivity patterns of the relations (symmetric/asymmetric, inverse, and composition) to a geometric interpretation such as rotation. Other approaches model the representations in higher dimensional space such as four-dimensional space (4D) to enhance the ability to infer the connectivity patterns (i.e., \textit{expressiveness}). However, modeling relation and entity in a 4D space often comes at the cost of interpretability.
This paper proposes HopfE, a novel KGE approach aiming to achieve interpretability of inferred relations in the 
four-dimensional space. We first model the structural embeddings in 3D Euclidean space and view the relation operator as an $SO(3)$ rotation. Next, we map the entity embedding vector from a 3D Euclidean space to a 4D hypersphere using the inverse Hopf Fibration, in which we embed the semantic information from the KG ontology. Thus, HopfE considers the structural and semantic properties of the entities without losing expressivity and interpretability. Our empirical results on four well-known benchmarks achieve state-of-the-art performance for the KG completion task.
\end{abstract}


\maketitle

\section{Introduction} \label{sec:introduction}
%Public KGs finds wide applicability in several downstream tasks such as entity linking, relation extraction, and fact-checking, and question answering \cite{ji2020survey}. 
%\todo[inline]{Botht the caption and sub-figures of Figure 1 are incomplete. I assume the given triple is (h,r,t) and it transformed to (k,?,j), this piece of knowledge should be added to the caption,  2) understanding the transformation from S2 or S3 to S3 or S4 is not straightforward, is there anyway to improve that?}

\begin{figure*}[t!]
  \begin{subfigure}[t]{0.20\textwidth}
    \includegraphics[width=\textwidth, viewport=0.1in -10 193 208, clip=false]{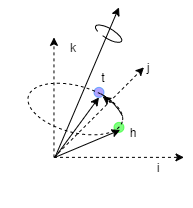}
    \caption{NagE}
    % \vspace{-2mm}
  \end{subfigure}%
  ~
  \begin{subfigure}[t]{0.20\textwidth}
    \includegraphics[width=\textwidth, viewport=0.1in 0 214 208, clip=false]{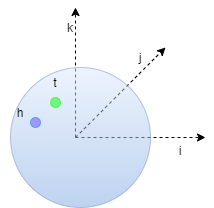}
    \caption{QuatE}
    % \label{fig:hist_q_inv}
    % \vspace{-2mm}
  \end{subfigure}
  ~
%   \hspace{2.1mm}
  \begin{subfigure}[t]{0.20\textwidth}
    \includegraphics[width=1.25\textwidth, viewport=-0.2in 0.1in 201 225, clip=false]{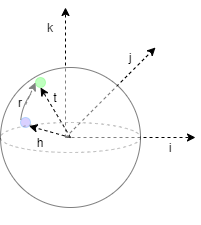}
    \caption{HopfE step 1: Rotation}
    % \label{fig:hist_q1_inv}
    % \vspace{-2mm}
  \end{subfigure}
  ~
%   \hspace{0.00mm}
  \begin{subfigure}[t]{0.20\textwidth}
    \includegraphics[width=\textwidth, viewport=-0.18in 0 210 217, clip=false]{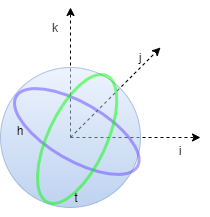}
    \caption{HopfE step 2: Fibration}
    % \label{fig:hist_q2_inv}
    % \vspace{-2mm}
  \end{subfigure}

  \caption{For a given triple <h,r,t> in three-dimensions (i,j,k); NagE provides geometric interpretability by rotating the entities in a 3D space (Figure (a)). Conversely, QuatE (Figure (b)) achieves higher order expressiveness by transforming the entities using the relational quaternion in 4D space, although it loses geometric interpretability. Here, the entities are stereo-graphically projected on a 3D sphere for understanding. HopfE first performs a rotation in 3D followed by a fibration to 4D as shown in Figures (c) and (d). This allows retaining geometric interpretability in 4D as entities are now depicted as fibers. In Figure (d), the stereo-graphic projection of two fibrations is shown before performing the rotations.}
  %\caption{(a) RotatE model rotates the entities in the complex plane. (b) QuatE transforms the entities using the relational quaternion in $S^3$. The entities are displayed on a stereographically projected sphere. (c) HopfE first performs a rotation in $S^2$ followed by a fibration to $S^3$ as in (d). The stereographic projection of the fibrations are shown before the rotations however we perform them after applying the relational operator and in an ideal case for a positive triple the two fibers would overlap. } Models achieving expressiveness in higher dimensions such as QuatE (Figure (b)), transforms the entities using the relational quaternion in $S^3$ though looses geometric interpretability. , unlike QuatE, where entities are points
  \label{fig:HopfEcomparision}
      \vspace{-3mm}
\end{figure*}

Publicly available Knowledge Graphs (KGs) (e.g., Freebase \cite{DBLP:conf/aaai/BollackerCT07} and Wikidata \cite{DBLP:conf/www/Vrandecic12}) expose interlinked representation of structured knowledge 
%introduce a new avenue of structured knowledge 
%\todo[inline]{SS: "introduce a new avenue of structured knowledge"--> expose interlinked representation of structured knowledge   }
for several information extraction tasks such as entity linking, relation extraction, and question answering \cite{bastos2020recon,sakor2020falcon,lu2019answering}. 
%Due to the changing dynamics of the world where new entities emerge and unknown relations about them are learned, KGs require continuous effort to be updated and maintained. 
However, publishing ever-growing KGs is often confronted with incompleteness, inconsistency, and inaccuracy \cite{ji2016knowledge}.
%However, the construction of massive scale KGs is often confronted with incompleteness and sparseness \cite{ji2016knowledge}. 
%\todo[inline]{construction typically means to create something from scratch while these KG are typically published or extracted from other formats }
%\todo[inline]{"However, the construction of massive scale KGs is often confronted with incompleteness and sparseness \cite{ji2016knowledge}."  --->
%However, publishing ever-growing KGs is often confronted with incompleteness, inconsistency, and inaccuracy. \cite{ji2016knowledge}.}
%\todo[inline]{what do you mean by "sparseness", I guess you mean incompleteness?? }
Knowledge Graph Embeddings (KGE) methods have been
proposed to overcome these challenges by learning entity and relation representations (i.e., embeddings) in vector spaces by capturing latent structures in the data \cite{wang2014knowledge}. 
%\todo[inline]{Knowledge Graph Embeddings (KGE) methods learn dense representations (so-called `embeddings´) for entities and relations by capturing latent structure and semantics of background knowledge \cite{wang2014knowledge}. or ´embeddings´The dense vectors can be employed in downstream learning tasks to overcome KG challenges.}
A typical KGE approach follows a two-step process \cite{gesese2019survey,wang2017knowledge}:
1) determining the form of entity and relation representations,
%2) defining a scoring function to calculate plausibility of a KG triple, and 3) learning entity and relation representations.
%\todo[inline]{form is not a clear term here, maybe behavior}
2) use a transformation function to learn entity and relation representations in a vector space for calculating the plausibility of KG triples via a scoring function.
%\todo[inline]{entities through relations might not be correct, whether both representation are leaned concurrently or individually?}
Intuitively, KGE approaches consider a set of original "connectivity patterns" in the KGs and learn representations that are approximately able to reason over the given patterns (\textit{expressiveness} property
of a KGE model \cite{wang2018multi}). 
%Intuitively, KGE methods infer the "connectivity patterns" in the KGs, and the ability to reason them is broadly considered as \textit{expressiveness}of a KGE model \cite{wang2018multi}. \todo[inline]{Intuitively, KGE approaches consider a set of original "connectivity patterns" in the KGs, and learn representations which are approximately able to reason over the given patterns (\textit{expressiveness} attributeof a KGE model \cite{wang2018multi}). }
%For example, few KG relations can be symmetric (e.g., sibling) or antisymmetric (e.g., father); someare the inverse of other relations (e.g., husband and wife); and few others are compositional relations (e.g., mother's brother is uncle).  $(h,r,t) -> (t,r,h)$
For example, in a connectivity pattern such as symmetric relations (e.g. sibling), anti-symmetric (e.g., father), inverse relations (e.g., husband and wife), or
compositional relations (e.g., mother's brother is uncle); \textit{expressiveness} means to hold these patterns in a vector space. 

\textbf{Limitations of state-of-the-art KGE Models}: 
%The initial KGE approaches (e.g., TransE \cite{bordes2013translating}) can infer specific connectivity patterns via the proposed scoring function using translation methods. 
%KGs are open, and largely explainable in terms of direct relationships between neighboring entities.
KGs are openly available, and in addition they are largely explainable because of the explicit semantics and structure (e.g., direct and labeled relationships between neighboring entities).
In contrast, KG embeddings are often referred to as sub-symbolic representation since they illustrate elements in the vector space, losing the original \textit{interpretability} deduced from logic or geometric interpretations \cite{palmonari2020knowledge}. Similar issue has been observed in a few initial KGE methods \cite{bordes2013translating,wang2014knowledge,lin2015learning} that lack \textit{interpretability}, i.e., these models fail to link the connectivity patterns to a geometric interpretation such as rotation \cite{sun2018rotate}. For example, in inverse relations, the angle of rotation is negative; in composition, the rotation angle is an addition ($r_1+r_2=r_3$), etc. Other KGE approaches such as RotatE \cite{sun2018rotate} and NagE \cite{yang2020nage} provide \textit{interpretability} by introducing rotations in two-dimensional (2D) and three-dimensional (3D) spaces respectively, to preserve the connectivity patterns. Conversely, QuatE \cite{DBLP:conf/nips/0007TYL19} imposes a transformation on the entity representations by quaternion multiplication with the relation representation. 
QuatE aims for higher dimensional \textit{expressiveness} by performing entity and relational transformation in hypercomplex space (4D in this case) \cite{DBLP:conf/nips/0007TYL19}. Quaternions enable expressive modeling in four-dimensional space and have more degree of freedom
than transformation in complex plane \cite{DBLP:conf/nips/0007TYL19,zhu2018quaternion,zhang2020beyond}.
%as quaternion multiplication 
Albeit expressive and empirically powerful, QuatE fails to provide explicit and understandable geometric interpretations due to its dependency on unit quaternion-based transformation \cite{yang2020nage,lu2020dense}. Hence, it is an open research question as to how can a KGE approach maintain the geometric \textit{interpretability} and simultaneously achieves the model's \textit{expressiveness} in four-dimensional space? 
%We aim to answer the concerned question.

%Hence, it remains an open question as how to fill the gap between the \textit{interpretability} of the model's rotations and its \textit{expressiveness} in four-dimensional space. In this paper, we aim to answer the following research question: can a KGE approach maintain the geometric \textit{interpretability} and simultaneously increase the model's \textit{expressiveness} in four-dimensional space?%\todo[inline]{4D is represented as $S^3$  or $S^4$ }
\textbf{Proposed Hypothesis:} The essential research gap to achieve interpretability in four-dimensional space for a KGE embedding motivates our work.
QuatE transforms one point (i.e., entity representation) to another in the 4D hypersphere $S^3$ \cite{DBLP:conf/nips/0007TYL19}. 
%\todo[inline]{This sentence does not make sense from 4D to 3D??? "QuatE transforms a point (i.e., entity representation) in the 4D hypersphere $S^3$ to another point in the $S^3$."}
Unlike QuatE, we hypothesize that it is viable to map a three-dimensional point to a circle
%\todo[inline]{what fiber means here? it is a term which has not been introduced before}
in the four-dimensional space. 
%\todo[inline]{why did you choose this hypothesis? just to unlike QuatE, you should dicuss what are the deficiencies of QuatE and based on which facts you place this hypothesis?}
Our rationale for the choices are: 1) the relational rotations in 3D sphere $S^2$ are already proven to be interpretable (i.e., \cite{yang2020nage,gao2020rotate3d}). However, transformations as done in QuatE lack geometric \textit{interpretability} \cite{lu2020dense}. Hence, performing relation rotations in $S^2$ will allow us to inherit these proven interpretable properties. 2) Mapping the entity representation of $S^2$ to a circle in a higher dimension can illustrate these representations as fibers \cite{whitney1940theory} (i.e., manifolds) 
%\todo[inline]{what fiber means here? it is a term which has not been introduced before}
for maintaining geometric interpretations while providing the \textit{expressiveness} of 4D embeddings (cf., Figure \ref{fig:HopfEcomparision}). In other words, by representing entities as fibers, we can aim for a direct mapping between rotations in $S^2$ to transformation of fibers in the $S^3$, without compromising \textit{interpretability} in four-dimensional space.
3) By converting a point of $S^2$ to a circle in $S^3$, we aim to semantically enrich entity embeddings through KG ontology (entity-aliases, entity-types, descriptions, and entity-labels as textual literals \cite{gesese2019survey}) onto the points on the fiber.

\textbf{Approach:} In this paper, we present HopfE- a novel KGE approach aiming to increase the \textit{expressiveness} in higher dimensional space while maintaining the geometric \textit{interpretability} (Figure \ref{fig:HopfEcomparision}). Specifically, in HopfE, we first rotate the relational embeddings in the 3D Euclidean space using Quaternion \cite{vicci2001quaternions} rotations. In the next step, the point (i.e., entity representation) in the 3D sphere is converted to a circle on the 4D hypersphere using inverse Hopf Fibration \cite{hopf1931abagrams}. We propose to use the extra dimension gained by Hopf mapping to also represent the entity attributes for semantically enriching the embeddings. \\
\textbf{Contributions}. The key contributions of this work are: 1)  Our work (theoretically and empirically) legitimizes Hopf Fibration's application in KG embeddings. To the best of our knowledge, we are the first to incorporate Hopf Fibration mapping for any representation learning approach. 2) We bridge the important gap between geometric \textit{interpretability} and four-dimensional \textit{expressiveness} of a knowledge graph embedding method. 3)  We provide a way to induce entity semantics in 4D, keeping the rotational interpretability intact. 4) Our empirical results achieves comparable state-of-the-art performance on four well-known KG completion datasets.

The structure of the paper is follows: Section \ref{sec:related} reviews the related work. Section \ref{sec:problem} formalizes the problem. Section \ref{sec:theoretical} describes theoretical foundations and Section  \ref{sec:method} describes our proposed approach.
Section \ref{sec:experiments} describes experiment setup. Our results are discussed in Section \ref{sec:results}. We conclude in Section \ref{sec:conclusion}.

%%%%%%%%%%%%%%%%%%%%%%%%%%%%%%%%%%%%%%%%%%%%%%%%%%%%%%%%%%%%%%%%%%%
%%%%%%%%%%%%%%%%%%%%%%%%%%%%%%%%%%%%%%%%%%%%%%%%%%%%%%%%%%%%%%%%%%%
%%%%%%%%%%%%%%%%%%%%%%%%%%%%%%%%%%%%%%%%%%%%%%%%%%%%%%%%%%%%%%%%%%%%%%%%%%%%%%%%%%%%%%%%%%%%%%%%%%%%%%%%%%%%%%%%%%%%%%%%%%%%%%%%%%%%%%%%%%%%%%%%%%%%%%%%%%%%%%%%%%%%%%%%%%%%%%%%%%%%%%%%%%%%%%%%%%%%%%%%%%%%%%%%%%%%%%%%%%%%%%%%%%%%%%%%%%%%%%%%%%%%%%%%%%%%%%%%%%%%%%%%%%%%%%%%%%%%%%%%%%%%%%%%%%%%%%%%%%%%%%%%%%%%%%%%%%%%%%%%%%%%%%%%%%%%%%%%%%%%%%%%%%%%%%%%%%%%%%%%%%%%%%%%%%%%%%%%%%%%%%%%%%%%%%%%%%%%%%%%%%%%%%%%%%%%%%%%%%%%%%%%%%%%%%%%%%%%%%%%%%%%%%%%%%%%%%%%%%%%%%%%%%%%%%%%%%%%%%%%%%%%%%%%%%%%%%%%%%%%%%%%%%%%%%%%%%%%%%%%%%%%%%%%%%%%%%%%%%%%%%%%%%%%%%%%%%%%%%%%%%%%%%%%%%%%%%%%%%%%%%%%%%%%%%%%%%%%%%%%%%%%%%%%%%%%%%%%%%%%%%%%%%%%%%%%%%%%%%%%%%%%%%%%%%%%%%%%%%%%%%%%%%%%%%
\section{Related Work} \label{sec:related}
Considering KGE is an extensively studied topic in the literature \cite{allen2021interpreting,zhao2020convolutional,ji2020survey,chami2020low,gao2020rotate3d}, we stick to the work closely related to our approach. Previous works on KG embeddings are primarily divided into translation and semantic matching models \cite{wang2017knowledge}. Translational models use distance-based scoring functions which are often based on learning a translation from the head entity to the tail entity. The first translation-based approach was TransE \cite{bordes2013translating} which models the composition of the entities and relations as $e_1 + r = e_2$. 
The second category (i.e., semantic matching) uses similarity-based scoring functions. The basic idea is to measure the plausibility of KG triples by matching the latent semantics of entities and relations expressed in their vector space representations. Distmult \cite{DBLP:journals/corr/YangYHGD14a} is a semantic matching model where the entities are represented as vectors. The relations are diagonal matrices, and the composition is a bilinear product of the entities and relations. ComplEx \cite{trouillon2016complex} is a modification of DistMult where the parameters are in the complex space. RotatE \cite{sun2018rotate} was the first to propose the translation as a rotation of the entities in the complex plane for geometric \textit{interpretability}. RotatE models the rotation in the 2D plane and has symmetric, antisymmetric, composition, and inversion properties.  This rotation, however, induces a strict compositional dependence between the relations. Yang et al. \cite{yang2020nage} introduced NagE, that provides a group theoretic perspective of rotations similar to RotatE by using Non-Abelian ($SO(3)$ and $SU2E$) groups. 
These groups give a geometrical interpretation of rotations in 3D, and being non-commutative groups, they solve the problem of strict compositionality.
QuatE \cite{DBLP:conf/nips/0007TYL19} has modeled the relational rotations in the quaternion four-dimensional space, thus making the model more expressive. QuatE also studied that increasing spatial dimensions such as to Octonion does not increase performance compared to modeling relation and entities in 4D. 
\textit{We position our work at the intersection of QuatE, and NagE approaches (cf. Figure \ref{fig:HopfEcomparision})}. 
In contrast with QuatE that does transformation using quaternion multiplication, we perform $SO(3)$ rotations using quaternion relation to inherit the geometric \textit{interpretability} of NagE. The inverse Hopf Fibration lets us perform the dimensional mapping which has found applications in diverse domains such as rigid body mechanics \cite{marsden1995introduction} and quantum information theory \cite{mosseri2001geometry}. Furthermore, there are extensive works for inducing schema and ontology information in the embedding models such as entity types \cite{zhao2020connecting,lin2016knowledge}, entity descriptions \cite{wang2016text}, and literals \cite{gesese2019survey}. We also aim to induce entity semantic properties (textual literals) into HopfE. Unlike existing methods \cite{gesese2019survey}, that concatenate the semantic and latent embeddings, we would like to look at the latent space for representing entity attributes.

%%%%%%%%%%%%%%%%%%%%%%%%%%%%%%%%%%%%%%%%%%%%%%%%%%%%%%%%%%%%%%%%%%%
%%%%%%%%%%%%%%%%%%%%%%%%%%%%%%%%%%%%%%%%%%%%%%%%%%%%%%%%%%%%%%%%%%%
%%%%%%%%%%%%%%%%%%%%%%%%%%%%%%%%%%%%%%%%%%%%%%%%%%%%%%%%%%%%%%%%%%%%%%%%%%%%%%%%%%%%%%%%%%%%%%%%%%%%%%%%%%%%%%%%%%%%%%%%%%%%%%%%%%%%%%%%%%%%%%%%%%%%%%%%%%%%%%%%%%%%%%%%%%%%%%%%%%%%%%%%%%%%%%%%%%%%%%%%%%%%%%%%%%%%%%%%%%%%%%%%%%%%%%%%%%%%%%%%%%%%%%%%%%%%%%%%%%%%%%%%%%%%%%%%%%%%%%%%%%%%%%%%%%%%%%%%%%%%%%%%%%%%%%%%%%%%%%%%%%%%%%%%%%%%%%%%%%%%%%%%%%%%%%%%%%%%%%%%%%%%%%%%%%%%%%%%%%%%%%%%%%%%%%%%%%%%%%%%%%%%%%%%%%%%%%%%%%%%%%%%%%%%%%%%%%%%%%%%%%%%%%%%%%%%%%%%%%%%%%%%%%%%%%%%%%%%%%%%%%%%%%%%%%%%%%%%%%%%%%%%%%%%%%%%%%%%%%%%%%%%%%%%%%%%%%%%%%%%%%%%%%%%%%%%%%%%%%%%%%%%%%%%%%%%%%%%%%%%%%%%%%%%%%%%%%%%%%%%%%%%%%%%%%%%%%%%%%%%%%%%%%%%%%%%%%%%%%%%%%%%%%%%%%%%%%%%%%%%%%%%%%%%%%

\section{Problem Formulation} \label{sec:problem}
% todo[inline]{Taken from RECON. Make changes to prevent plagiarism?}
We define a KG as a tuple $KG = (\mathcal{E},\mathcal{R},\mathcal{T}^+)$ where $\mathcal{E}$ denotes the set of entities (vertices), $\mathcal{R}$ is the set of relations (edges), and $\mathcal{T}^+ \subseteq \mathcal{E} \times \mathcal{R} \times \mathcal{E} $ is a set of all triples.
A triple $\uptau = (e_h,r_{ht},e_t) \in \mathcal{T}^+$ indicates that, for the relation $r_{ht} \in \mathcal{R}$, $e_h$ is the head entity (origin of the relation) while $e_t$ is the tail entity. Since $KG$ is a multigraph; $e_h=e_t$ may hold and $|\{r_{e_h,e_t}\}|\geq 0$ for any two entities. We define the tuple $(A^e,\tau^e) = \varphi(e)$ obtained from a context retrieval function $\varphi$, that returns, for any given entity $e$, the set: $A^e$ which is a set of all attributes of the entity such as aliases, descriptions, and type. $\tau^e \subset \mathcal{T}^+$ is the set of all triples with head at $e$.
%\todo[inline]{you did not define $\tau^e$}
The \textit{KB completion task} (KBC) predicts the entity pairs $\langle e_i,e_j\rangle$ in the Knowledge Graph that have a relation $r_{ij} \in \mathcal{R}$ between them.  
%\todo[inline]{Rephrase: \textit{ The KB completion task} (KBC) predicts a relation $r_{ij} \in \mathcal{R}$  between a given entity pair $\langle e_i,e_j\rangle$ of the Knowledge Graph.}
Similar to other researchers~\cite{DBLP:conf/nips/0007TYL19,yang2020nage}, we view KBC as a ranking task. 
%\todo[inline]{add one more supportive sentence to make it clear}
In addition, we also aim to model KG contextual information to improve the ranking. 
%This is achieved by learning representations of the set $A^e$ as described in section \ref{sec:method}.

%%%%%%%%%%%%%%%%%%%%%%%%%%%%%%%%%%%%%%%%%%%%%%%%%%%%%%%%%%%%%%%%%%%%%%%%%%%%%%%%%%%%%%%%%%%%%%%%%%%%%%%%%%%%%%%%%%%%%%%%%%%%%%%%%%%%%%%%%%%%%%%%%%%%%%%%%%%%%%%%%%%%%%%%%%%%%%%%%%%%%%%%%%%%%%%%%%%%%%%%%%%%%%%%%%%%%%%%%%%%%%%%%%%%%%%%%%%%%%%%%%%%%%%%%%%%%%%%%%%%%%%%%%%%%%%%%%%%%%%%%%%%%%%%%%%%%%%%%%%%%%%%%%%%%%%%%%%%%%%%%%%%%%%%%%%%%%%
%%%%%%%%%%%%%%%%%%%%%%%%%%%%%%%%%%%%%%%%%%%%%%%%%%%%%%%%%%%%%%%%%%%
%%%%%%%%%%%%%%%%%%%%%%%%%%%%%%%%%%%%%%%%%%%%%%%%%%%%%%%%%%%%%%%%%%%
%%%%%%%%%%%%%%%%%%%%%%%%%%%%%%%%%%%%%%%%%%%%%%%%%%%%%%%%%%%%%%%%%%%%%%%%%%%%%%%%%%%%%%%%%%%%%%%%%%%%%%%%%%%%%%%%%%%%%%%%%%%%%%%%%%%%%%

\section{Theoretical Foundations} \label{sec:theoretical}
\subsection{Quaternions and Rotations in 3D}
Quaternions \cite{hamilton1844lxxviii} are a numerical representation of vectors in four dimensions. 
%Quaternions \cite{hamilton1844lxxviii} are a number system representing vectors in four dimensions. \todo[inline]{Rephrase: Quaternions \cite{hamilton1844lxxviii} are a numerical representation of vectors in four dimensions. }
It can be considered as an ordered tuple $(a,b,c,d)$ of numbers represented by $a + bi + cj + dk$ where $a$ is the real component and $b,c,d$ are the imaginary components.
%\todo[inline]{what a real and imaginary component means?}
%The multiplication rules for the unit basis vectors $i,j,k$ can be encapsulated as
%\begin{align*}
 %   i^2 = j^2 = k^2 = -1 \\
  %  ij=k jk=i ki=j  \\
   % ji=-k kj=-i ik=-j
%\end{align*}
The norm \cite{pugh2002real} of the quaternion $a+bi+cj+dk$ can be represented by $\sqrt{a^2+b^2+c^2+d^2}$. From hereon, "$*$" is multiplication operation, "$\odot$" is an element-wise product, "$\otimes$" denotes Hamilton product.
%Now, we sketch the result for rotation in 3D sphere below:
In the following, we describe the details of rotation in 3D sphere.

\begin{theorem}\label{thm_so3}
For a pure quaternion $v=p_1i+p_2j+p_3k$ in 3D that has real components as zero, the quaternion $r=a+bi+cj+dk$ with unit norm gives a rotational map $rvr^{-1}$, where the axis of rotation denoted by bi+cj+dk, and the angle of rotation by $2\arccos{a}$
\end{theorem}
%todo[inline]{The Theorem is very unclear, what "the scalar component is taken to be 0" means?}
\begin{proof}
 We divide the proof into three sections where we prove that the norm of $v$ is preserved, where $(b,c,d)$ is the eigen vector of the rotational map, and the angle of rotation can be given in terms of $a$. The rotational map $R_r(v) = rvr^{-1}$ has the norm $\norm{R_r(v)} = \norm{rvr^{-1}} = \norm{r}*\norm{v}*\norm{r^{-1}}$. We have assumed $r$ to have unit norm, it's inverse will also have unit norm. Hence, the norm of $v$ is maintained. To prove that $b,c,d$ is the eigen vector of the rotational map, it is sufficient to show that the rotation of $bi+cj+dk$ under the rotation map results in a scaled version of the vector:
% \begin{equation}
\begin{align*}
    R_r(bi+cj+dk) &= (a + bi + cj + dk)\otimes\\
    &(bi + cj + dk)\otimes\frac{(a - bi - cj - dk)}{a^2+b^2+c^2+d^2} \\
    &= bi + cj + dk
\end{align*}
%\b{equation}
%Hence we get the same vector with scaling operation.
Hence we obtain the same vector with scaling operation. Now we prove that the angle of rotation can be given in terms of $a$. We can consider $a$ vector perpendicular to the axis vector $bi + cj + dk$. Let $w = ci - bj$ be this vector without loss of generality. Then, it's rotation is given by:
\begin{align*}
    R_r(w) &= (a + bi + cj + dk)\otimes(ci - bj)\\
    &\otimes(a - bi - cj -dk) \\
        %   &= ((ac+db)i + (-ad+dc)j - (b^2+c^2)k)\\
        %   &= ((-b^2+c^2) + (ab+cd)i + (-ac+db)j + (-2bc)k) \\
           &= ((ac+db)i + (-ab+dc)j - (b^2+c^2)k)\\
           &\otimes(a - bi - cj -dk) \\
          &= (c*(a^2 - b^2 -c^2 - d^2) + 2adb)i + \\
          &(-b*(a^2 - b^2 -c^2 - d^2) + 2adc)j + (-2ab^2 -2ac^2)k
\end{align*}
The cosine of the angle between w and $R_r(w)$ is given by
\begin{align*}
    \cos{\theta} &= \frac{w.R_r(w)}{\lVert w \lVert ^2} \\
                 &= a^2 - b^2 - c^2 -d^2 = 2a^2 - 1 \\
          \theta &= 2\arccos{a}   
\end{align*}
If $b=c=0$, we take $w=i$ to get the same result.
Thus, the rotational map $a + bi + cj + dk$ encodes the axis and angle of rotation.
\end{proof}

%%%%%%%%%%%%%%%%%%%%%%%%%%%%%%%%%%%%%%%%%%%%%%%%%%%%%%%%%%%%%%%%%%%
%%%%%%%%%%%%%%%%%%%%%%%%%%%%%%%%%%%%%%%%%%%%%%%%%%%%%%%%%%%%%%%%%%%
%%%%%%%%%%%%%%%%%%%%%%%%%%%%%%%%%%%%%%%%%%%%%%%%%%%%%%%%%%%%%%%%%%%
\subsection{Hopf Fibration}
Two spaces are said to be homotopic if there exists a set of continuous deformations from one space to another \cite{baues1989algebraic}.
%\todo[inline]{Two spaces behave as homotopic when a set of continuous deformations from one space to another exists \cite{baues1989algebraic}.}
%These deformations could be a shrinking, bending, or stretching of the space \cite{baues1989algebraic}. 
For example, a solid circular disc is homotopic to a point as it could be deformed along the radial lines to a point. %The set of deformations in space are called homotopy.
\begin{definition}[Homotopy Lifting]
Given a map $\pi: E \xrightarrow[]{} B$ between two topological spaces $E$ and $B$ and a topological space $X$, we say that $\pi$ has a homotopy lifting property with respect to $X$ if for a homotopy $g: X \times [0,1] \xrightarrow[]{} B$, there exists a homotopy $f: X \times [0,1] \xrightarrow[]{} E$ such that $\pi \circ f = g$.
\end{definition}
\begin{definition}[Fibration]
A fibration is a map $f: E \xrightarrow[]{} B$ between two topological spaces $E$ and $B$ such that the homotopy lifting property is satisfied for all spaces.
\end{definition}
First introduced in 1931, Hopf Fibration is a mapping from $S^3$ to $S^2$ \cite{hopf1931abagrams,lyons2003elementary}. For the quaternion $r = a + bi + cj + dk$ this map is given by
\begin{equation}
    M(r) = (a^2+b^2-c^2-d^2, 2(ad+bc), 2(bd-ac))
\end{equation}
This is equivalent to applying the rotational map $R_r$ to the point (1,0,0). The Hopf map ($M(r)$) assigns all $r$ in $S^3$ to $P$ in $S^2$, such that $R_r$ would rotate $(1,0,0)$ to $P$.
 We could also verify that multiplying $e^{it}=cos(t)+i*sin(t)$ to this map doesn't change the rotation. The inverse Hopf map gives the pre-image of point $P$ by:
\begin{lemma}\label{hopf_lemma}
For a point $P = (p1,p2,p3)$ the inverse Hopf map is given by $M^{-1}(P) = r'e^{it}$ where $0 \leq t \leq 2\pi$ and $r'$ can be given by one of the below equations:
\begin{equation}\label{eq_inv_hopf_map_1}
  r' = \frac{1}{\sqrt{2(1+p_1)}}((1+p_1)i + p_2j + p_3k)
\end{equation}
\centering{OR}
\begin{equation}\label{eq_inv_hopf_map_2}
  r' = \frac{\sqrt{1+p_1}}{2}*(1-\frac{p_3i}{1+p_1}+\frac{p_2k}{1+p_1})
\end{equation}
% $r = \frac{(1+p1)i + p2j + p3k}{\sqrt{2(1+p1)}}$ and 
\end{lemma}
\begin{proof}
% Since $h(r) = r \otimes i \otimes r^{-1} = (re^{it}) \otimes i \otimes (re^{it})^{-1}$, it can be verified by substituting for r that $h(r) = p_1i + p_2j + p_3k$ and hence the proof.
We begin by noting that (from Theorem \ref{thm_so3}),
\begin{align*}
    M(r'e^{it}) = (r'e^{it}) \otimes i \otimes (r'e^{it})^{-1} 
\end{align*}
Substituting for $r'e^{it}$ from Equation \ref{eq_inv_hopf_map_1} we get,
\begin{align*}
    r' \otimes e^{it} &= \frac{1}{\sqrt{2(1+p_1)}}((1+p_1)i + p_2j + p_3k) \otimes\\
    &(\cos{t} + i*\sin{t})) \\
    &= \frac{1}{\sqrt{2(1+p_1)}}*(\sin{t}*(1+p_1) + \\
    &\cos{t}*(1+p_1))i + (p_2*\cos{t}+p_3*\sin{t})j +\\
    &(p_3*\cos{t}-p_2*\sin{t})k)
\end{align*}
Similarly we get the compliment $\overline{r'e^{it}}$ as,
\begin{align*}
    \overline{r' \otimes e^{it}} &= \frac{1}{\sqrt{2(1+p_1)}}*(\sin{t}*(1+p_1) \\
    & -\cos{t}*(1+p_1))i - (p_2*\cos{t}+p_3*\sin{t})j +\\
    & -(p_3*\cos{t}-p_2*\sin{t})k)
\end{align*}
The Hopf Map of the set of points on $r' \otimes e^{it}$ is the point in $S^2$ that is obtained by rotation of $(1,0,0)$ using any of these quaternions,
\begin{align*}
    M(r'e^{it}) &= (r'e^it) \otimes (i) \otimes (\overline{r'e^it}) \\
    &= p_1i + p_2j + p_3k = P
\end{align*}
therefore, $M^{-1}{(P)} = r'e^{it}$. (also applicable to Equation \ref{eq_inv_hopf_map_2}).
\end{proof}
A point on the 3D sphere is mapped to a circle by the Equation \ref{eq_inv_hopf_map_1}. Furthermore, a set of points is mapped to a group of linked rings, and a circle in $S^2$ is mapped to a torus in $S^3$ as can be seen in the Figure \ref{fig:hopf_maps}. The property of linked circles is stated in Lemma \ref{linked_circles}.
\begin{lemma}\label{linked_circles}
Every pair of fiber formed by the Hopf Map are linked to each other forming a set of linked circles.
\end{lemma}

\begin{figure*}[htbp]
  \centering
% Was this the ultimate question of life, the universe and everything?--->maybe not :)
% \setkeys{Gin}{width=0.42\linewidth,height=0.42\linewidth}
\setkeys{Gin}{width=0.16\linewidth,height=0.15\linewidth}
  \includegraphics{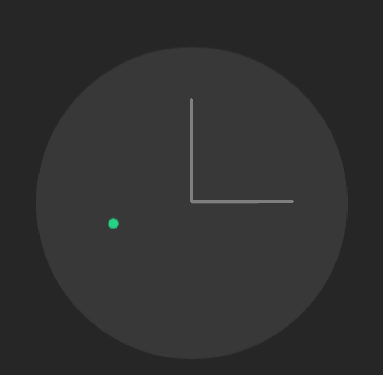}\,%
  \includegraphics{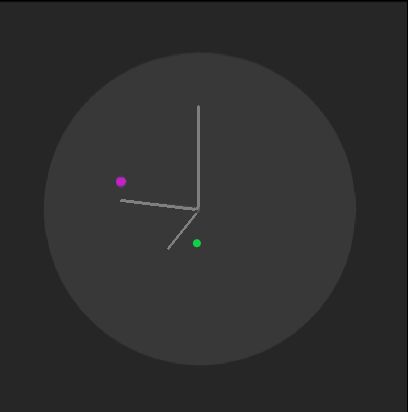}\,%
  \includegraphics{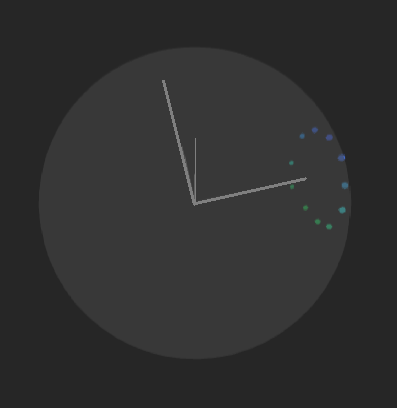}\,%
  \includegraphics{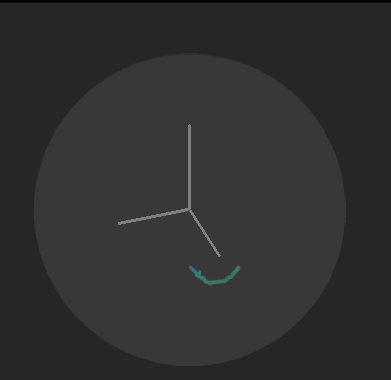}\,%
  \includegraphics{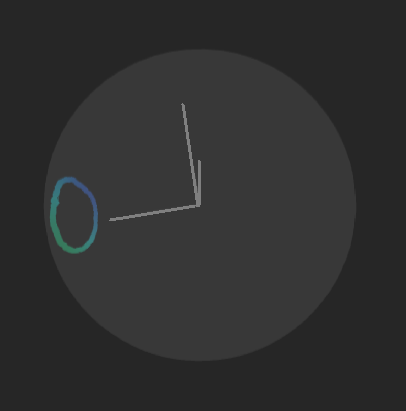}

  \includegraphics{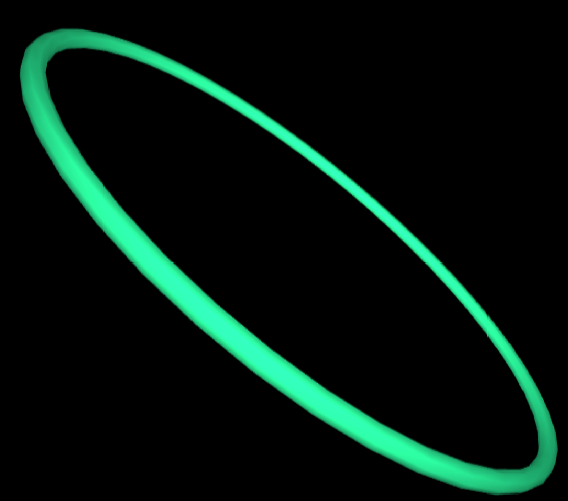}\,%
  \includegraphics{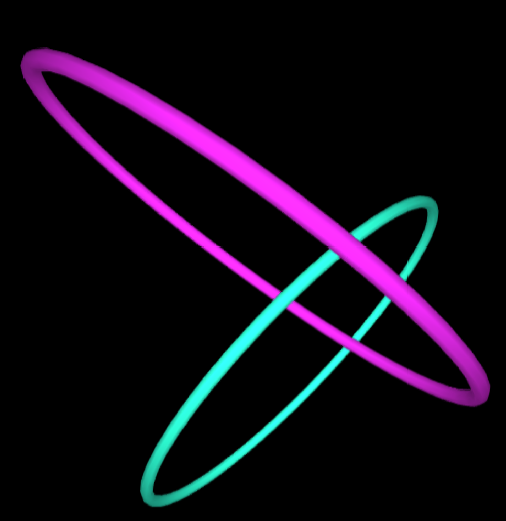}\,%
  \includegraphics{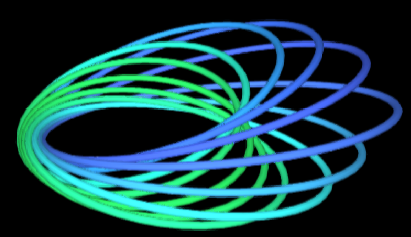}\,%
  \includegraphics{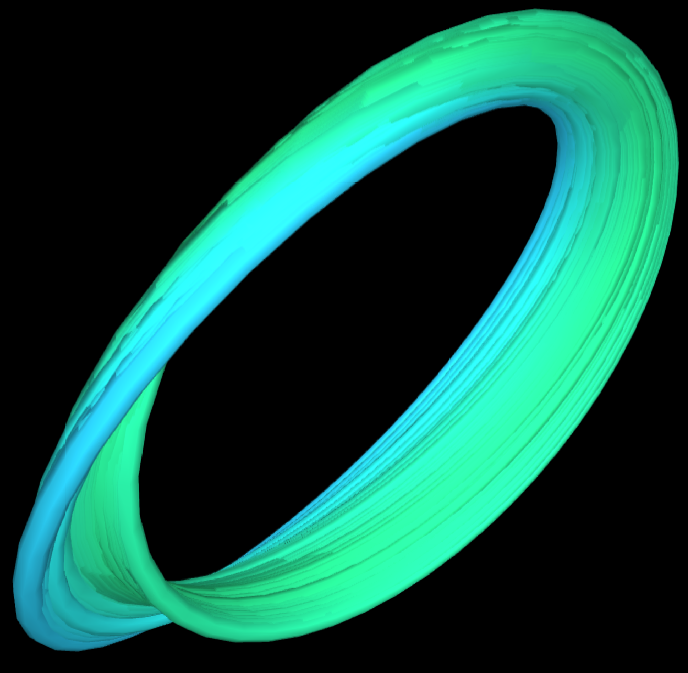}\,%
  \includegraphics{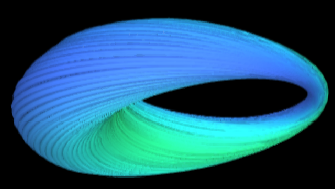}
  
  \includegraphics{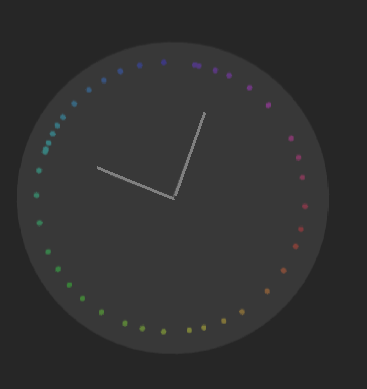}\,%
  \includegraphics{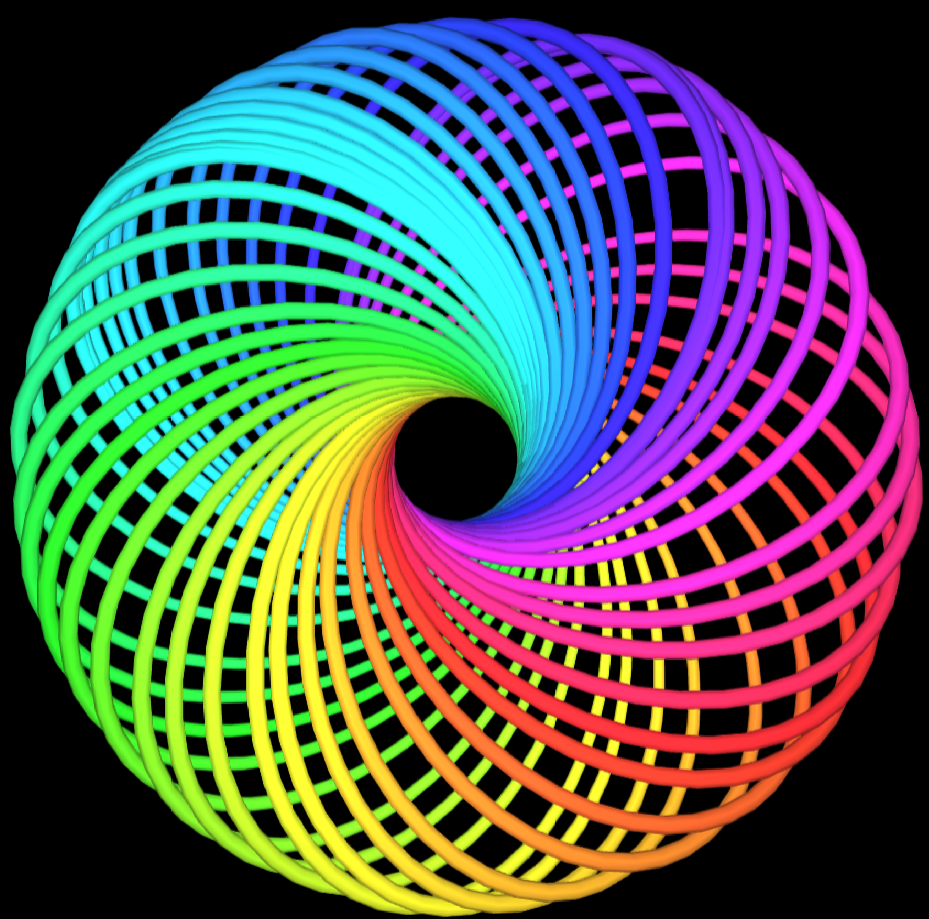}\,%
  \includegraphics{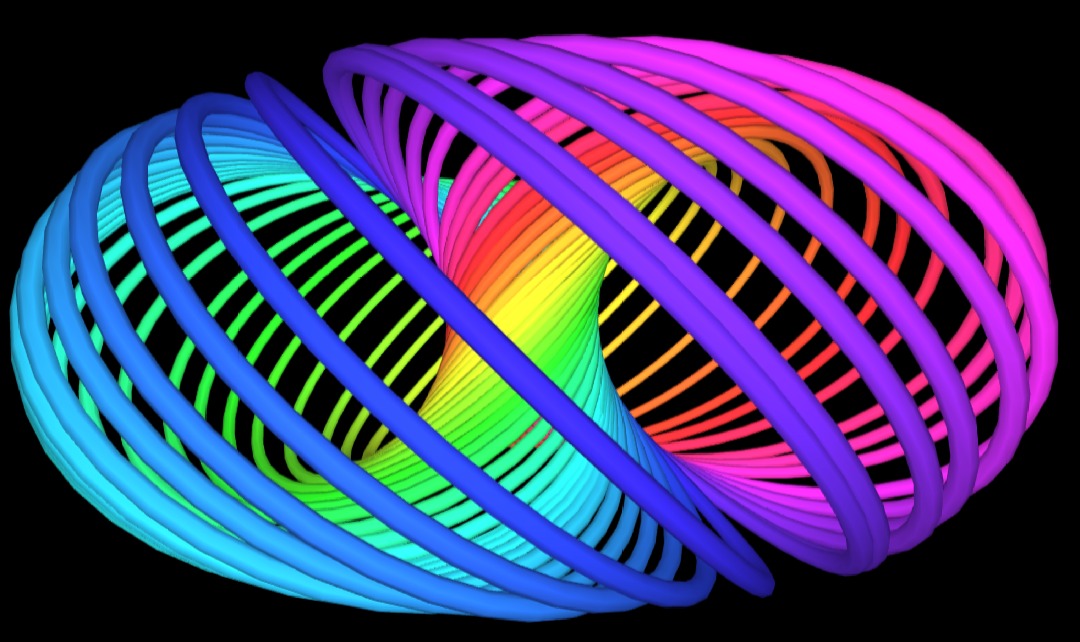}\,%
  \includegraphics{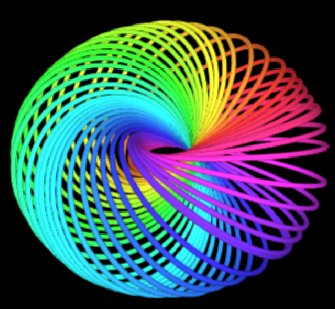} 
 \caption{We show the Hopf fibers of points on $S^2$ visualized \cite{visualization} using the stereo-graphic projection from $S^3$ to $S^2$. The first row from the top illustrates (left to right): a point, two points, discrete points, arc, and circle in 3D space. The second row (left to right) illustrates the stereographic projection of corresponding fibers in 3D. The third row (left to right) shows the set of points on the grid circle of $S^2$ and its corresponding top-view, side-view, and isometric views, respectively. }
  %\caption{Hopf fibers of points on $S^2$ visualized using the stereo-graphic projection from $S^3$ to $S^2$. The first row shows the points and the second row shows the corresponding fibers. The third row shows the top, side and isometric views of the set of points on the diameter of $S^2$}
  \label{fig:hopf_maps}
\end{figure*}
%\footnotetext{Hopf maps for visualizaation created using \url{https://samuelj.li/hopf-fibration/}}

%%%%%%%%%%%%%%%%%%%%%%%%%%%%%%%%%%%%%%%%%%%%%%%%%%%%%%%%%%%%%%%%%%%
%%%%%%%%%%%%%%%%%%%%%%%%%%%%%%%%%%%%%%%%%%%%%%%%%%%%%%%%%%%%%%%%%%%
%%%%%%%%%%%%%%%%%%%%%%%%%%%%%%%%%%%%%%%%%%%%%%%%%%%%%%%%%%%%%%%%%%%
%%%%%%%%%%%%%%%%%%%%%%%%%%%%%%%%%%%%%%%%%%%%%%%%%%%%%%%%%%%%%%%%%%%%%%%%%%%%%%%%%%%%%%%%%%%%%%%%%%%%%%%%%%%%%%%%%%%%%%%%%%%%%%%%%%%%%%
\subsection{Optimal Transport}
%The optimal matching between two sets can be found using the well studied problem of optimal transport 
The aim of optimal transport \cite{monge1781memoire} is to minimize the cost of transport from one (probability) measure to another. The well known Wasserstein Distance \cite{villani2009a} defines the optimal transport plan to move an amount of matter from one location to another. We rely on the principle of optimal transport considering we aim to map entity representation from one plane to another hyperplane. 
\begin{definition}[]
Let $p \in [1,\infty)$ and c: $R^n \times R^n$ $\xrightarrow{} [0,\infty)$ be the cost of transporting the measure $\mu$ to $\nu$, then the $p^{th}$ Wasserstein distance between the measures is given by
\begin{equation} \label{wasserstein}
    W_p(\mu,\nu) = \inf_{\gamma \in \Pi(\mu,\nu)}\left(\int_{R^n \times R^n} c(x,y)^p \partial{\gamma}\right)^{\frac{1}{p}}
\end{equation}
where $\Pi$ is the transport plan between the measures $\mu$ and $\nu$.
\end{definition}

\begin{theorem}\label{thm_min_dist}
The minimum of the euclidean distance between the points of the Hopf fibers of two points in the $\delta$ vicinity in $S^2$ is bounded from above by $\delta^2$. 
\end{theorem}
\begin{proof}
Consider a point $P=(p_1,p_2,p_3)$ in $S^2$ and another point $P^{'}=(p_1^{'},p_2^{'},p_3^{'})$ such that \[\lVert P-P^{'}\rVert=\delta\]
From Lemma \ref{hopf_lemma} we know that the the fiber of $P$ is given by \[\frac{1}{\sqrt{2(1+p_1)}}((1+p_1)i+p_2j+p_3k) \otimes e^{it}\] and the fiber of $P^{'}$ is given by \[\frac{1}{\sqrt{2(1+p_1^{'})}}((1+p_1^{'})i+p_2^{'}j+p_3^{'}k) \otimes e^{it^{'}}\], where $t$ and $t^{'}$ are the respective parameters. 
The square of the euclidean distance ($D$) between any two points of the fibers is given by the $l_2$ norm as below
\begin{align*}
    D &= (-k*\sin(t)(1+p_1)+k^{'}sin(t^{'})(a+p_1^{'}))^2 \\
    &= (k*\cos(t)(1+p_1)-k^{'}sin(t^{'})(a+p_1^{'}))^2 \\
    &= (k*(\cos(t)p_2+\sin(t)p_3)-k^{'}(cos(t^{'})p_2^{'}+sin(t^{'})p_3^{'}))^2 \\
    &= (k*(\cos(t)p_3-\sin(t)p_2)-k^{'}(cos(t^{'})p_3^{'}+sin(t^{'})p_2^{'}))^2 \\
\end{align*}
We then proceed by taking the derivative of $D$ with respect to $t$ and $t^{'}$, equating them with $0$ for the optimal solution. Then, we have \[\frac{\partial{f}}{\partial{t}}=0, \frac{\partial{f}}{\partial{t^{'}}}=0\]
%Then, we have \[\frac{\partial{f}}{\partial{t}}=0\] and \[\frac{\partial{f}}{\partial{t^{'}}}=0\]
Solving for $t$ and $t^{'}$ at $\delta \approx 0$, gives \[t=\frac{\delta p_1(p_3-p_2)}{(1+p_1)(p_3^{2}+p_2^{2})},   t^{'}=\frac{\delta (p_3-p_2)}{(1+p_1)(p_3^{2}+p_2^{2})}\]
%Solving for $t$ and $t^{'}$ at $\delta \approx 0$, gives \[t=\frac{\delta p_1(p_3-p_2)}{(1+p_1)(p_3^{2}+p_2^{2})}\] and \[t^{'}=\frac{\delta (p_3-p_2)}{(1+p_1)(p_3^{2}+p_2^{2})}\]
Substituting these values in $D$ we get $D \leq \mathcal{O}(\delta^{2})$.
\end{proof}
From Theorem \ref{thm_min_dist} we can deduce that the minimum distance between the fibers is a monotone increasing function of the distance between their corresponding points. Therefore:
\begin{corollary}\label{4d_as_exp_as_3d}
For any set of entities, if there is a configuration satisfying the relations between them in $S^2$- then there will also be a configuration on the fibers that satisfy the concerning relations.
\end{corollary}
Hence, from Corollary  \ref{4d_as_exp_as_3d}, the triples expressed in the lower dimension are the same as in the higher dimension after the mapping. Now we want to show that the higher dimensional space can represent more triples. We state the corresponding Lemma: 
\begin{lemma}\label{lemma_neq}
The triple set represented by rotations in $S^2$ is not equal to the set of triples that can be represented in $S^3$ after the fibration. 
\end{lemma}
\begin{proof}
Let $\tau^2$ be the set of triples that can be represented in $S^2$ and $\tau^3$ be the set of triples that can be expressed in $S^3$ after the Fibration(mapping). We have to show that $\tau^2 \nsubseteq \tau^3$. For this, it is enough to show that even a single configuration of triples in $\tau^3$ is not in $\tau^2$. 
Consider three entities $e_1$, $e_2$ and $e_3$ with relation $r_1$ existing between them such that ${(e_1,r_1,e_2),(e_2,r_1,e_3),(e_1,r,e_3)} \in \tau^{+}$ and ${(e_3,r_1,e_2),(e_2,r_1,e_1),(e_3,r_1,e_1)} \in \tau^{-}$ where $\tau^{+}$ and $\tau^{-}$ are the set of positive and negative triples respectively. Let $E_1, E_2, E_3 \in R^3$ be the representations of these entities in $S^2$ and let $R_1 \in H$ ($H$ is Hamiltonian Space) be the relation quaternion. Thus we have,

%\begin{align*}
% $ R_1 \otimes E_1 \otimes \overline{R_1} = E_2; R_1 \otimes E_2 \otimes \overline{R_1}= E_3;   R_1 \otimes E_1 \otimes \overline{R_1} = E_3$
$ R_1 \otimes E_1 \otimes R_1^{-1} = E_2; R_1 \otimes E_2 \otimes R_1^{-1}= E_3;   R_1 \otimes E_1 \otimes R_1^{-1} = E_3$
%\end{align*}

%\begin{align*}
 %   R \otimes E_1 \otimes \overline{R} &= E_2 \\
  %  R \otimes E_2 \otimes \overline{R} &= E_3 \\
   % R \otimes E_1 \otimes \overline{R} &= E_3 \\
%\end{align*}
This implies that $E_1=E_2=E_3$ and the axis of rotation is along these vectors. However, then we would also have: 

% $R_1 \otimes E_2 \otimes \overline{R_1} = E_1;  R_1 \otimes E_3 \otimes \overline{R_1} = E_2;   R_1 \otimes E_3 \otimes \overline{R_1} = E_1$
$R_1 \otimes E_2 \otimes R_1^{-1} = E_1;  R_1 \otimes E_3 \otimes R_1^{-1} = E_2;   R_1 \otimes E_3 \otimes R_1^{-1} = E_1$

%\begin{align*}
 %   R \otimes E_2 \otimes \overline{R} &= E_1 \\
  %  R \otimes E_3 \otimes \overline{R} &= E_2 \\
   % R \otimes E_3 \otimes \overline{R} &= E_1 \\
%\end{align*}
which is a contradiction, so the relational operator in $S^2$ cannot represent the triples' set. In $S^3$, we have each of these entities' fibers coinciding with each other from the above analysis. We could now select $A_1^1, A_1^2 \in C$, $A_2^1, A_2^2 \in C$ and $A_3^1 \in C$ as the attributes for $e_1$, $e_2$ and $e_3$ respectively. We define $R^c(\in C)=e^{i\theta}$ to be the rotational operator in the complex plane $C$ for the relation. Setting $A_1^1=e^{i\theta_1}$, $A_1^2=A_2^1=e^{i\theta_2}$ and $A_2^2=A_3^1=e^{i\theta_3}$, where $\theta_2-\theta_1=\theta_3-\theta_2=\theta$, we can see that the set of triples (both positive and negative) can be represented by the mapping in $S^3$.
\end{proof}
\begin{theorem}\label{lemma_strictsubset}
The set of triples that can be represented by rotations in $S^2$ is a strict subset of the set of triples that can be represented in $S^3$ after the fibration. 
\end{theorem}
\begin{proof}
Let $\tau^2$ be the set of triples that can be represented in $S^2$ and $\tau^3$ be the set of triples that can be expressed in $S^3$ after the fibration(mapping). We have to show that $\tau^2 \subset \tau^3$. From \ref{4d_as_exp_as_3d} we have $\tau^2 \subseteq \tau^3$ and from \ref{lemma_neq} we have $\tau^2 \nsubseteq \tau^3$. Thus, we can deduce that $\tau^2 \subset \tau^3$ and hence the proof.
\end{proof}
%%%%%%%%%%%%%%%%%%%%%%%%%%%%%%%%%%%%%%%%%%%%%%%%%%%%%%%%%%%%%%%%%%%%%%%%%%%%%%%%%%%%%%%%%%%%%%%%%%%%%%%%%%%%%%%%%%%%%%%%%%%%%%%%%%%%%%%%%%%%%%%%%%%%%%%%%%%%%%%%%%%%%%%%%%%%%%%%%%%%%%%%%%%%%%%%%%%%%%%%%%%%%%%%%%%%%%%%%%%%%%%%%%%%%%%%%%%%%%%%%%%%%%%%%%%%%%%%%%%%%%%%%%%%%%%%%%%%%%%%%%%%%%%%%%%%%%%%%%%%%%%%%%%%%%%%%%%%%%%%%%%%%%%%%%%%%%%%%%%%%%%%%%%%%%%%%%%%%%%%%%%%%%%%%%%%%%%%%%%%%%%%%%%%%%%%%%%%%%%%%%%%%%%%%%%%%%%%%%%%%%%%%%%%%%%%%%%%%%%%%%%%%%%%%%%%%%%%%%%%%%%%%%%%%%%%%%%%%%%%%%%%%%%%%%%%%%%%%%%%%%%%%%%%%%%%%%%%%%%%%%%%%%%%%%%%%%%%%%%%%%%%%%%%%%%%%%%%%%%%%%%%%%%%%%%%%%%%%%%%%%%%%%%%%%%%%%%%%%%%%%%%%%%%%%%%%%%%%%%%%%%%%%%%%%%%%%%%%%%%%%%%%%%%%%%%%%%%%%%%%%%%%%%%%%%%%%%%%%%%%%%%%%%%%%%%%%%%%%%%%%%%%%%%%%%%%%%%%%%%%%%%%%%%%%%%%%%%%%%%%%%%%%%%%%%%%%%%%%%%%%%%%%%%%%%%%%%%%%%%%%%%%%%%%%%%%%%%%%%%%%%%%%%%%%%%%%%%%%%%%%%%%%%%%%%%%%%%%%%%%%%%%%%%%%%%%%%%%%%%%%%%%%%%%%%%%%%%%%%%%%%%%%%%%%%%%%%%%%%%%%%%%%%%%%%%%%%%%%%%%%%%%%%%%%%%%%%%%%%%%%%%%%%%%%%%%%%%%%%%%%%%%%%%%%%%%%%%%%%%%%%%%%%%%%%%%%%%%%%%%%%%%%%%%%%%%%%%%%%%%%%%%%%%%%%%%%%%%%%%%%%%%%%%%%%%%%%%%%%%%%%%%%%%%%%%%%%%%%%%%%%%%%%%%%%%%%%%%%%%%%%%%%%%%%%%%%%%%%%%%%%%%%%%%%%%%%%%%%%
%%%%%%%%%%%%%%%%%%%%%%%%%%%%%%%%%%%%%%%%%%%%%%%%%%%%%%%%%%%%%%%%%%%
%%%%%%%%%%%%%%%%%%%%%%%%%%%%%%%%%%%%%%%%%%%%%%%%%%%%%%%%%%%%%%%%%%%
%%%%%%%%%%%%%%%%%%%%%%%%%%%%%%%%%%%%%%%%%%%%%%%%%%%%%%%%%%%%%%%%%%%
\section{Our Approach} \label{sec:method}
\subsection{Modeling rotations using Quaternions} \label{sec:approach-rotation}
As noted in the Theorem \ref{thm_so3}, the unit quaternion is isomorphic to the $SU(2)$ group which in turn is homomorphic to the $SO(3)$ rotation group (also observed by NagE \cite{yang2020nage}). This allows us to use the quaternion for rotations in the 3D Euclidean space as below:
\begin{equation}  \label{rotatedentity}
    w^{'} = q \otimes w \otimes \overline{q}
\end{equation}
where $q$ is the quaternion, $\overline{q}$ is the inverse of $q$, $w$ is the vector in 3D space and $w^{'}$ is the rotated vector. In other words,  $q$ is the dimension-wise relation quaternion, $w$ is the dimension-wise entity representation in 3D space, and $w^{'}$ is the rotated entity representation in the same space. 
If $q = a + bi + cj + dk$ then $\overline{q} = \frac{a - bi -cj -dk}{\sqrt{a^2 + b^2 + c^2 + d^2}}$.
%We use the relation quaternion on the entity vectors per dimension for the rotation in three-dimensional space exploited in the previous literature. 

Reusing $SO(3)$ rotations from NagE in this step helps us inherit its geometric for the completion of our approach.
% Mention the SO3 rotation equations above
% define relational rotation
%%%%%%%%%%%%%%%%%%%%%%%%%%%%%%%%%%%%%%%%%%%%%%%%%%%%%%%%%%%%%%%%%%%%%%%%%%%%%%%%%%%%%%%%%%%%%%%%%%%%%%%%%%%%%%%%%%%%%%%%%%%%%%%%%%%%%%%%%%%%%%%%%%%%%%%%%%%%%%%%%%%%%%%%%%%%%%%%%%%%%%%%%%%%%%%%%%%%%%%%%
%%%%%%%%%%%%%%%%%%%%%%%%%%%%%%%%%%%%%%%%%%%%%%%%%%%%%%%%%%%%%%%%%%%
%%%%%%%%%%%%%%%%%%%%%%%%%%%%%%%%%%%%%%%%%%%%%%%%%%%%%%%%%%%%%%%%%%%
%%%%%%%%%%%%%%%%%%%%%%%%%%%%%%%%%%%%%%%%%%%%%%%%%%%%%%%%%%%%%%%%%%%
%%%%%%%%%%%%%%%%%%%%%%%%%%%%%%%%%%%%%%%%%%%%%%%%%%%%%%%%%%%%%%%%%%%
\begin{table*}[htb]
\small 
    \centering
    \begin{tabular}{p{2.0cm}p{6.0cm}p{6.0cm}p{1.5cm}}
        %\Cline{1-4}
        \toprule
        %& & micro & \\
        % & & & \\
        % Model & \multicolumn{5}{c|}{FB15K237}& \multicolumn{5}{c}{YAGO3-10} \\
        Model & Scoring function & Parameters & $\mathcal{O}_{time}$ \\
        %\Cline{1-4}
        \midrule
        TransE \cite{bordes2013translating}& $\lVert E_h + W_r - E_t \rVert$ &  $E_h, W_r, E_t \in R^k$ & $\mathcal{O}(k)$ \\
      %  HolE& $\langle W_r, Qh \star Qt \rangle$ & $Qh, Wr, Qt \in R^k$ & $\mathcal(klog(k))$ \\
        DistMult \cite{DBLP:journals/corr/YangYHGD14a}& $\langle W_r, E_h, E_t \rangle$ & $E_h, W_r, E_t \in R^k$ & $\mathcal{O}(k)$ \\
        ComplEx \cite{trouillon2016complex}& $Re(\langle W_r, E_h, \overline{E_t} \rangle)$ & $E_h, W_r, E_t \in C^k$ & $\mathcal{O}(k)$ \\
        RotatE \cite{sun2018rotate}& $\lVert E_h \odot Wr - E_t \rVert$ & $E_h, W_r, E_t \in C^k, \mid W_r \mid = 1$ & $\mathcal{O}(k)$ \\
        QuatE \cite{DBLP:conf/nips/0007TYL19}& $E_h \otimes W_r \cdot E_t$ & $E_h, W_r, E_t \in H^k$ & $\mathcal{O}(k)$ \\
        \midrule
        HopfE (ours)& $\lVert M(W_r \otimes E_h \otimes \overline{W_r}, A^h) - M(E_t, A^t)\rVert$ & $E_h, E_t \in (R^3)^k, W_r \in H^k, A^h,A^t \in C^k$ & $\mathcal{O}(k)$ \\
        \bottomrule
    \end{tabular}
    \caption{KGE Overview. $W_r$ is the relation representation, $E_h$, $E_t$ are the entity representations, $M$ is the inverse Hopf map. }
    \label{tab:tab-complexity}
        \vspace{-2mm}
\end{table*}
\subsection{Homotopy Lifting using Inverse Hopf Map}\label{sec:approach-hopf}
In this step, 
%we focused on maintaining geometric \textit{interpretability}. 
we aim to achieve \textit{expressiveness} in higher dimension space, i.e., 4D hypersphere as our main contribution. The inverse map with the equations described in Lemma \ref{hopf_lemma} lifts the point (rotated entity representation from Equation \ref{rotatedentity}) in $S^2$ to a set of points (i.e., fiber) lying in a higher-dimensional space. Once both entities are mapped as individual fibers in the 4D hypersphere, it is possible to ground connectivity patterns in rotations as done by NagE and RotatE in the lower dimensional spaces. Each point on the fiber has the same geometric property for a given entity. We parametrize the entities to fix the location of these sets of points on the fiber. Hence, we now have a bunch of points on the fiber for each entity, which could represent the entity's semantic attributes. Formally we define a set of parameters $t$ for each entity corresponding to its semantic attribute $A^e$ (cf., section \ref{sec:problem}). These parameters could be used to find the set of points on the fiber as below:
\begin{equation}\label{eq_parameter}
    e_i = r' \otimes (\cos(t) + \sin(t)i)
\end{equation}
where $e_i$ is the entity quaternion per dimension and $r'$ is obtained from either Equations \ref{eq_inv_hopf_map_1} or \ref{eq_inv_hopf_map_2}. From Corollary \ref{4d_as_exp_as_3d}, we obtain that the equation \ref{eq_parameter}, which takes the representation to a higher dimension, does not cause a loss in expressivity.
\begin{claim}\label{claim_nclique}
Translation methods, in the euclidean space, require at least $n-1$ dimensions to model an n-clique relational pattern between $n$ entities.
\end{claim}
A direct consequence of Claim \ref{claim_nclique} is that to model dense relations; we would need a dimension proportionate to the number of entities. We propose an increase in the model expressivity by increasing the entity attribute size while keeping the dimensions same. The entity attribute size (i.e. number of heads) is $\mid A^{e} \mid$. While increasing the number of heads will make the model more expressive, it may also increasing the training latency. 
In ablation results, we aim to study the relation between model expressivity and the attribute heads for a fixed time interval (cf., section \ref{sec:heads}).
%%%%%%%%%%%%%%%%%%%%%%%%%%%%%%%%%%%%%%%%%%%%%%%%%%%%%%%%%%%%%%%%%%%
%%%%%%%%%%%%%%%%%%%%%%%%%%%%%%%%%%%%%%%%%%%%%%%%%%%%%%%%%%%%%%%%%%%
%%%%%%%%%%%%%%%%%%%%%%%%%%%%%%%%%%%%%%%%%%%%%%%%%%%%%%%%%%%%%%%%%%%
%%%%%%%%%%%%%%%%%%%%%%%%%%%%%%%%%%%%%%%%%%%%%%%%%%%%%%%%%%%%%%%%%%%%%%%%%%%%%%%%%%%%%%%%%%%%%%%%%%%%%%%%%%%%%%%%%%%%%%%%%%%%%%%%%%%%%%%%%%%%%%%%%%%%%%%%%%%%%%%%%%%%%%%%%%%%%%%%%%%%%%%%%%%%%%%%%%%%%%%%%
\subsection{Inducing Semantics} \label{sec:semantics}
%From Lemma \ref{hopf_lemma}, the inverse Hopf Map would take us from a point in $S^2$ to a circular fiber in $S^3$.
%Homotopy lifting using inverse Hopf Fibration in the previous section would take us from a point in $S^2$ to a circular fiber in $S^3$. 
Once the point is mapped to the fiber using inverse Hopf Fibration, we propose to learn a set of points that would inform about the various semantics/attributes of the entity. One possible way for the operation would be to have the points spread out on the circle (fiber) based on the semantic attributes of the entity. Taking inspiration from \cite{gesese2019survey,zhao2020connecting} we use the textual literals widely available for an entity in the standard KGs namely: entity description, Instance of (type of entity such as human), Alias (alternative names) and Surface forms (label of the entity). Specifically we use the CBOW model \cite{mikolov2013efficient} to aggregate the embeddings of the tokens in each semantic attribute. These embeddings could be initialized from a standard model such as Glove \cite{pennington2014glove}. The parameter $t$ in Equation \ref{eq_parameter} is learned from the semantics. We also tried with other methods such as initializing the parameters with the semantic embeddings and using an MLP to learn the final representation after concatenating the structural and semantic embeddings, which does not perform better.
%%%%%%%%%%%%%%%%%%%%%%%%%%%%%%%%%%%%%%%%%%%%%%%%%%%%%%%%%%%%%%%%%%%
%%%%%%%%%%%%%%%%%%%%%%%%%%%%%%%%%%%%%%%%%%%%%%%%%%%%%%%%%%%%%%%%%%%%%%%%%%%%%%%%%%%%%%%%%%%%%%%%%%%%%%%%%%%%%%%%%%%%%%%%%%%%%%%%%%%%%%%%%%%%%%%%%%%%%%%%%%%%%%%%%%%%%%%%%%%%%%%%%%%%%%%%%%%%%%%%%%%%%%%%%%%%%%%%%%%%%%%%%%%%%%%%%%%%%%%%%%%%%%%%%%%%%%%%%%%%%%%%%%%%%%%%%%%%%%%%%%%%%%%%%%%%%%%%%%%%%%%%%%%%%%%%%%%%%%%%%%%%%%%%%%%%%%%%%%%%%%%%%%%%%%%%%%%%%%%%%%%%%%%%%%%%%%%%%%%%%%%%%%%%%%%%%%%%%%%%%%%%%%%%%%%%%%%%%%%%%%%%%%%%%%%%%%%%%%%%%%%%%%%%%%%%%%%%%%%%%%%%%%%%%%%%%%%%%
%%%%%%%%%%%%%%%%%%%%%%%%%%%%%%%%%%%%%%%%%%%%%%%%%%%%%%%%%%%%%%%%%%%%%%%%%%%%%%%%%%%%%%%%%%%%%%%%%%%%%%%%%%%%%%%%%%%%%%%%%%%%%%%%%%%%%%
%%%%%%%%%%%%%%%%%%%%%%%%%%%%%%%%%%%%%%%%%%%%%%%%%%%%%%%%%%%%%%%%%%%
\begin{table*}[!htb]
    \centering
    \begin{tabular}{p{2.5cm}|p{1.0cm}p{1.0cm}p{1.0cm}p{1.0cm}p{1.0cm}|p{1.0cm}p{1.0cm}p{1.0cm}p{1.0cm}p{1.0cm}}
        %\Cline{1-4}
        \toprule
        %& & micro & \\
        % & & & \\
        Model & \multicolumn{5}{c|}{WN18}& \multicolumn{5}{c}{WN18RR} \\
         & MR & MRR & Hits@1 & Hits@3 & Hits@10 & MR & MRR & Hits@1 & Hits@3 & Hits@10 \\
        %\Cline{1-4}
        \midrule
        TransE \cite{bordes2013translating}&  -&  0.495& 0.113& 0.888& 0.943& 3384& 0.226& -& -& 0.501  \\
        DistMult \cite{DBLP:journals/corr/YangYHGD14a}& 655& 0.797& -& -& 0.946& 5110& 0.43& 0.39& 0.44& 0.49  \\
        ComplEx \cite{trouillon2016complex} & -& 0.941& 0.936& 0.945& 0.947& 5261& 0.44& 0.41& 0.46& 0.51   \\ 
        RotatE \cite{sun2018rotate}& 309& 0.949& \textbf{0.944}& 0.952& 0.959& 3340& 0.476& 0.428& 0.492& 0.571   \\
        NagE \cite{yang2020nage}& -& \textbf{0.950}& \textbf{0.944}& \underline{0.953}& \underline{0.960}& -& 0.477& 0.432& \underline{0.493}& 0.574  \\
        QuatE \cite{DBLP:conf/nips/0007TYL19}& 388& \underline{0.949}&\underline{0.941}& \textbf{0.954}& \underline{0.960}& 3472& \underline{0.481}& \textbf{0.436}&\textbf{0.500}& 0.564  \\
        \midrule
        HopfE & \underline{245}& \underline{0.949}& 0.938& \textbf{0.954}& \underline{0.960}&\underline{2885} & 0.472& 0.413&\textbf{0.500}& \textbf{0.586} \\
        HopfE + Semantics & \textbf{99}& 0.945& 0.934& 0.951& \textbf{0.961}& \textbf{2884}& \textbf{0.482}& \underline{0.433}& \textbf{0.500}& \underline{0.579}  \\
        \bottomrule
    \end{tabular}
    \caption{Evaluation metrics on the WN18 and WN18RR datasets. Best results are in bold and second best is underlined.}
    \label{tab:tab-wn}
        \vspace{-2mm}
\end{table*} 

\subsection{Score function and Optimization Objective} \label{sec:approach-optimization}
In previous sections, we learned that lifting the dimensions of an entity would take us to the fiber (circle) in $S^3$. For two entities we would have two linked circles enriched with entity semantics. Now, we also need to find the distance between the attributes of two entities lying on the linked circles. This could be done, for example, by considering the corresponding attributes of the two entities or by taking a pairwise distance between the all combinations of the attributes. That would be computationally expensive. Hence, we propose another variant where we find an optimal matching between the attributes and then take the pairwise distance between these matched attributes. Table \ref{tab:tab-complexity} provides an holistic overview of KGE details. To find an optimal matching between the set of entity attributes, we use the Wasserstein Distance using Equation \ref{wasserstein}. In practice, we have used the Sinkhorn iterations for faster computation of the Wasserstein distances \cite{cuturi2013sinkhorn}. Let $\Pi$ be the optimal matching between the attributes $\emph{h}_{i}$ $\in$ $\emph{E}_h$ and $\emph{t}_{j}$ $\in$ $\emph{E}_t$ of an entity pair $h$ and $t$ respectively. The score function is\footnote{$\emph{E}_h$ and $\emph{E}_t$ are the multi set corresponding to $e_h$ and $e_t$ in higher dimensions}:
% {\small
 \begin{equation} \label{score_fn}
 \small
    f_r(h,r_{ht},t) = \sum_{[i,j] \in \Pi} \frac{1}{2} \left( \lVert M(\mathcal{O} (\emph{h}_i,r_{ht}))-M(\emph{t}_j)\rVert + \lVert M(\mathcal{O} (\emph{t}_j,\overline{r}))-M(\emph{h}_i)\rVert \right)
    \end{equation}
% }
% \begin{multline}\label{eq_prob_sep_space_attn}
% \small
%     P(r \mid h, t, s) =  softmax(\\\mlp([r_{v_i,v_j}\parallel\vec{e}_h\parallel\vec{e}_t\parallel\vec{e}_h^{~attn}\parallel\vec{e}_t^{~attn}))
% \end{multline}

%\begin{multline}\label{score_fn}
%\small
 %   f_r(h,r_{ht},t) = \sum_{[i,j] \in \Pi} & \frac{1}{2} \left( \lVert Map(\mathcal{O} (\emph{h}_i,r_{ht}))-Map(\emph{t}_j)\rVert + \lVert Map(\mathcal{O} (\emph{t}_j,\overline{r}))-Map(\emph{h}_i)\rVert \right)
%\end{multline}

where $\emph{h}_i$, $\emph{t}_j$ are the attributes (semantic properties) of the head and tail entities respectively, $r_{ht}$ is the relation. $\mathcal{O}$ represents the rotation operation on each dimension of the entity i.e. $\mathcal{O}(e_d,r_d) = r_d \otimes e_d \otimes \overline{r_d}$ where $r_d$ $\in$ $H$ and $e_d \in R^3$. Similarly, $M$ represents the inverse Hopf Map to the fiber in $S^3$, per dimension. In practice, for larger attribute sizes, computing the Wasserstein distance between a set of points could be computationally expensive. Therefore, we define its variant to consider the minimum distance between the two sets' attributes. The score function transforms to:
{\small
\begin{equation}\label{score_fn2}
\small
    f_r(h,r_{ht},t) = \underset{\underset{\emph{t}_{j} \in \emph{E}_t}{\emph{h}_{i} \in \emph{E}_h} }{\inf} \frac{1}{2} \left( \lVert M(\mathcal{O} (\emph{h}_i,r_{ht}))-M(\emph{t}_j)\rVert + \lVert M(\mathcal{O} (\emph{t}_j,\overline{r_{ht}}))-M(\emph{h}_i)\rVert \right)
\end{equation}
}
The loss function aims to minimize the score function for the positive triples and maximize it for the negative triples. More specifically, we use a self adverserial negative sampling loss similar to RotatE \cite{sun2018rotate} as below:
\begin{align*}
\small
    L &= \sum_{h,r_{ht},t \in \tau} -\log \left\{ \sigma(\gamma + f_r(h,r_{ht},t)) \right\} \\
    - & \sum_{h^{-},r_{ht},t^{-} \in \tau^{-}} p(h^{-},r_{ht},t^{-}) \log \left\{ \sigma(-(\gamma + f_r(h^{-},r_{ht}^{-},t^{-})))  \right\}
\end{align*}
where $\tau$ is the set of positive triples and $\tau^{-}$ is the set of negative triples, $\gamma$ is a margin hyperparamter and $p$ is the probability given by the model of the negative sample.

\begin{table}[ht!]
\small
    \centering
    \begin{tabular}{p{1.19cm}|p{0.98cm}p{0.98cm}p{0.80cm}p{0.80cm}p{0.80cm}p{0.80cm}}
        %\Cline{1-4}
        \toprule
        %& & micro & \\
        % & & & \\
        % Model & \multicolumn{5}{c|}{FB15K237}& \multicolumn{5}{c}{YAGO3-10} \\
        Dataset & \#Entities & \#Relations & \% N to N & \% N to 1 & \% 1 to 1 & \% 1 to N \\
        %\Cline{1-4}
        \midrule
        WN18& 40943& 36&  0.00& 59.65& 40.22& 0.12 \\
        WN18RR& 40943& 22& 0.03& 74.59& 25.33& 0.04 \\
        FB15K237& 14541& 474& 4.93& 91.26& 3.63& 0.17 \\
        YAGO310& 123182& 74& 57.12& 41.10& 1.65& 0.12 \\
        \bottomrule
    \end{tabular}
    \caption{Dataset Stats including fraction of relation types.}
    \label{tab:tab-dataset_stats}
        \vspace{-2mm}
\end{table} 
%%%%%%%%%%%%%%%%%%%%%%%%%%%%%%%%%%%%%%%%%%%%%%%%%%%%%%%%%%%%%%%%%%%%%%%%%%%%%%%%%%%%%%%%%%%%%%%%%%%%%%%%%%%%%%%%%%%%%%%%%%%%%%%%%%%%%%%%%%%%%%%%%%%%%%%%%%%%%%%%%%%%%%%%%%%%%%%%%%%%%%%%%%%%%%%%%%%%%%%%%%%%%%%%%%%%%%%%%%%%%%%%%%%%%%%%%%%%%%%%%%%%%%%%%%%%%%%%%%%%%%%%%%%%%%%%%%%%%%%%%%%%%%%%%%%%%%%%%%%%%%%%%%%%%%%%%%%%%%%%%%%%%%%%%%%%%%%%%%%%%%%%%%%%%%%%%%%%%%%%%%%%%%%%%%%%%%%%%%%%%%%%%%%%%%%%%%%%%%%%%%%%%%%%%%%%%%%%%%%%%%%%%%%%%%%%%%%%%%%%%%%%%%%%%%%%%%%%%%%%%%%%%%%%%%%%%%%%%%%%%%%%%%%%%%%%%%%%%%%%%%%%%%%%%%%%%%%%%%%%%%%%%%%%%%%%%%%%%%%%%%%%%%%%%%%%%%%%%%%%%%%%%%%%%%%%%%%%%%%%%%%%%%%%%%%%%%%%%%%%%%%%%%%%%%%%%%%%%%%%%%%%%%%%%%%%%%%%%%%%%%%%%%%%%%%%%%%%%%%%%%%%%%%%%%%%%%%%%%%%%%%%%%%%%%%%%%%%%%%%%%%%%%%%%%%%%%%%%%%%%%%%%%%%%%%%%%%%%%%%%%%%%%%%%%%%%%%%%%%%%%%%%%%%%%%%%%%%%%%%%%%%%%%%%%%%%%%%%%%%%%%%%%%%%%%%%%%%%%%%%%%%%%%%%%%%%%%%%%%%%%%%%%%%%%%%%%%%%%%%%%%%%%%%%%%%%%%%%%%%%%%%%%%%%%%%%%%%%%%%%%%%%%%%%%%%%%%%%%%%%%%%%%%%%%%%%%%%%%%%%%%%%%%%%%%%%%%%%%%%%%%%%%%%%%%%%%%%%%%%%%%%%%%%%%%%%%%%%%%%%%%%%%%%%%%%%%%%%%%%%%%%%%%%%%%%%%%%%%%%%%%%%%%%%%%%%%%%%%%%%%%%%%%%%%%%%%%%%%%%%%%%%%%%%%%%%%%%%%%%%%%%%%%%%%%%%%%%%%%%%%%%%%%%%%%%%%%%%%%
%%%%%%%%%%%%%%%%%%%%%%%%%%%%%%%%%%%%%%%%%%%%%%%%%%%%%%%%%%%%%%%%%%%
%%%%%%%%%%%%%%%%%%%%%%%%%%%%%%%%%%%%%%%%%%%%%%%%%%%%%%%%%%%%%%%%%%%
%%%%%%%%%%%%%%%%%%%%%%%%%%%%%%%%%%%%%%%%%%%%%%%%%%%%%%%%%%%%%%%%%%%

\section{Experiment Setup} \label{sec:experiments}
\textbf{Datasets}: We executed experiments on the widely used benchmark datasets of WN18, WN18RR, FB15k237 and YAGO3-10 \cite{sun2018rotate,toutanova2015observed}. The WN18 dataset is extracted from Wordnet which contains lexical relations between English words. WN18RR removes the inverse relations in the WN18 dataset. FB15k237 was created from the Freebase knowledge base and does not contain inverse relation. The YAGO3-10 dataset is obtained from Wikipedia, Wikidata, and Geonames. Table \ref{tab:tab-dataset_stats} summarizes detailed statistics of all datasets.

\begin{table*}[!htb]
    \centering
    \begin{tabular}{p{2.5cm}|p{1.0cm}p{1.0cm}p{1.0cm}p{1.0cm}p{1.0cm}|p{1.0cm}p{1.0cm}p{1.0cm}p{1.0cm}p{1.0cm}}
        %\Cline{1-4}
        \toprule
        %& & micro & \\
        % & & & \\
        Model & \multicolumn{5}{c|}{FB15K237}& \multicolumn{5}{c}{YAGO3-10} \\
         & MR & MRR & Hits@1 & Hits@3 & Hits@10 & MR & MRR & Hits@1 & Hits@3 & Hits@10 \\
        %\Cline{1-4}
        \midrule
        TransE \cite{bordes2013translating} &  357& 0.294& -& -& 0.465& -& -& -& -& -  \\
         DistMult \cite{DBLP:journals/corr/YangYHGD14a} & 254& 0.241& 0.155& 0.263& 0.419& 5926& 0.340& 0.240& 0.380& 0.540  \\
        ComplEx \cite{trouillon2016complex} & 339& 0.247& 0.158& 0.275& 0.428& 6351& 0.360& 0.260& 0.400& 0.550   \\ 
        RotatE \cite{sun2018rotate} & \underline{177}& 0.338& 0.241& 0.375& \underline{0.533}& 1767& 0.495& 0.402& 0.550& 0.670   \\
        NagE \cite{yang2020nage} & -& \underline{0.340}& \underline{0.244}& \underline{0.378}& 0.530& -& -& -& -& -  \\
        QuatE \cite{DBLP:conf/nips/0007TYL19} & \textbf{176}& 0.311& 0.221& 0.342& 0.495& -& -& -& -& -  \\
        \midrule
        HopfE & 212&\textbf{0.343} & \textbf{0.247}& \textbf{0.379}&\textbf{0.534}& \textbf{1077}& \textbf{0.529}&\textbf{0.438} & \textbf{0.586}&\textbf{0.695} \\
        HopfE + Semantics & 199& 0.330& 0.230& 0.373& 0.528& \underline{1260}& \underline{0.518}& \underline{0.429}& \underline{0.570}& \underline{0.678}  \\
        \bottomrule
    \end{tabular}
    \caption{Evaluation metrics on the FB15K237 and YAGO3-10 datasets. Best results are in bold and second best is underlined.}
    \label{tab:tab-fb_yago}
        \vspace{-2mm}
\end{table*} 
\noindent \textbf{Evaluation Metrics:}
Following the widely adapted metrics, we report the Mean Rank (MR), Mean Reciprocal Rank (MRR) and the cut-off hit ratio (hits@N (N=1,3,10)) metrics on all the datasets. MR
reports the average rank in all the correct entities. MRR is
the average of the inverse rank of the correct entities. H@N evaluates the ratio of correct
entities predictions at a top N predicted results.
\begin{table}[ht!]
    \centering
    \begin{tabular}{p{2.5cm}|p{1.0cm}p{1.0cm}p{1.0cm}p{1.0cm}}
        %\Cline{1-4}
        \toprule
        %& & micro & \\
        % & & & \\
        % Model & \multicolumn{5}{c|}{FB15K237}& \multicolumn{5}{c}{YAGO3-10} \\
        Model & MR & MRR & Hits@1 & Hits@10 \\
        %\Cline{1-4}
        \midrule
        HopfE + Label & 1260& 0.518& 0.429& 0.678 \\
        HopfE + Alias & 1281& \textbf{0.525}& \textbf{0.44}& \textbf{0.684} \\
        HopfE + Instance & \textbf{1215}& 0.518& 0.428&  0.68 \\
        HopfE + Desc & 1264& 0.521& 0.436& 0.677 \\
        HopfE + Numeric & 2261& 0.441& 0.357& 0.598 \\
        \bottomrule
    \end{tabular}
    \caption{Ablation study of inducing individual literals on the YAGO3-10 dataset. Best results are in bold.}
    \label{tab:tab-ablation_semantics}
        \vspace{-2mm}
\end{table} 
\begin{table}[ht!]
    \centering
    \begin{tabular}{p{3.5cm}|p{1.0cm}|p{1.0cm}}
        \toprule
        \textbf{Relation}& \textbf{RotatE}& \textbf{HopfE} \\
        \midrule
        derivationally\_related\_form& 0.947& \textbf{0.956} \\
        also\_see& 0.585& \textbf{0.626} \\
        verb\_group& \textbf{0.943}& \textbf{0.943} \\
        similar\_to& \textbf{1}& \textbf{1} \\
        hypernym& 0.148& \textbf{0.161} \\
        instance\_hypernym& 0.318& \textbf{0.347} \\
        member\_meronym& 0.232& \textbf{0.257} \\
        synset\_domain\_topic\_of& 0.341& \textbf{0.405 }\\
        has\_part& 0.184& \textbf{0.199} \\
        member\_of\_domain\_usage& 0.318& \textbf{0.34} \\
        member\_of\_domain\_region& 0.2& \textbf{0.392} \\
        \bottomrule
    \end{tabular}
    \caption{MRR of individual relations of the WN18RR dataset.}
    \label{tab:tab-rel_wise_study}
        \vspace{-2mm}
\end{table} 
\noindent \textbf{Implementation:} The Adam optimizer \cite{DBLP:journals/corr/KingmaB14} is used for learning with a learning rate of 0.1 with an exponentially decaying learning schedule with a decay rate of 0.1. The parameters are initialized using the He initialization \cite{he2015delving}. We perform a hyperparameter search on the embedding dimension $d \in \{50,100,200,500\}$, batch size $b \in \{256,512,1024\}$, negative samples $n_{neg} \in \{16,32,62,128,256,512\}$, negative adversarial sampling temperature $\alpha \in \{0.5, 1.0\}$ and the margin paramter $\gamma \in \{6,9,12,24\}$. We release code in public Github\footnote{\url{https://github.com/ansonb/HopfE}}. We consider two variants of our approach: 1) \textit{HopfE}, that does not contain any semantic properties 2) \textit{HopfE+Semantics}, which has entity attributes encoded on the fibers (cf., section \ref{sec:semantics}). \\
\textbf{Baselines:} We compare with the state-of-the-art link prediction methods mainly employing translational/semantic matching techniques since our approach is built upon them. Specifically we compare with TransE \cite{bordes2013translating}, Distmult \cite{DBLP:journals/corr/YangYHGD14a}, ComplEx \cite{trouillon2016complex}, RotatE \cite{sun2018rotate}, NagE \cite{yang2020nage}, and QuatE \cite{DBLP:conf/nips/0007TYL19}. We extend 3D rotations similar to NagE, however, our hyper parameters are restricted due to hardware limitations.
We inherit results and experiment settings from \cite{yang2020nage,DBLP:conf/nips/0007TYL19} for WN18, WN18RR, and FB15K237. We ignore alternative settings and hyper parameters provided by \cite{sun2020re,ruffinelli2019you}. On the Yago dataset, results are taken from RotatE. To keep the experiment settings same, the configuration of QuatE with N3 regularization and reciprocal learning has been omitted from our comparison considering these are optimizations extraneous to the model and we do not use them. 
\section{Results} \label{sec:results}
Our evaluation results are shown in the Table \ref{tab:tab-wn} and Table \ref{tab:tab-fb_yago}. On the WN18 dataset, the majority relation patterns comprise symmetry/antisymmetry and inversion. Our model which combines both properties (geometric \textit{interpretability} and 4D \textit{expressiveness}) achieves comparable results to the baselines. On the WN18RR dataset, the main relation pattern is the symmetry pattern, and the translation method (TransE) shows a limited performance because of its inability to infer symmetric patterns. QuatE reports slightly higher results than NagE, and our model achieves the best results on all metrics (except Hits@1). Similar to WN18RR, FB15K237 dataset does not contain inverse relation, and the majority relation pattern is composition. We see that NagE performs better than QuatE on the FB15K237 dataset due to its ability to infer non-abelian hyper relations. Except for MR, our model achieves better results than previous baselines illustrating the ability of HopfE in handling composition relations. On the YAGO3-10 dataset, triples mostly contain a human's descriptive attributes, such as marital status, profession, citizenship, and gender. Our model (HopfE) clearly outperforms the previous baseline (RotatE) on all metrics. Across all datasets, we observe that adding semantics does not necessarily result in a gain in the performance and the behavior varies as per the dataset.
\begin{table*}[!htb]
    \centering
    \begin{tabular}{p{2.5cm}|p{1.0cm}p{1.0cm}|p{1.0cm}p{1.0cm}|p{1.0cm}p{1.0cm}|p{1.0cm}p{1.0cm}}
        %\Cline{1-4}
        \toprule
        Model & \multicolumn{2}{c|}{N to N}& \multicolumn{2}{c|}{N to 1}& \multicolumn{2}{c|}{1 to N}& \multicolumn{2}{c}{1 to 1} \\
         & MRR & Hits@1 & MRR & Hits@1 & MRR & Hits@1 & MRR & Hits@1 \\
        %\Cline{1-4}
        \midrule
        HopfE(w/o Hopf) & 0.917& 0.889& 0.187& 0.112& \textbf{0.259}&\textbf{ 0.175}& 0.865& 0.762  \\
        \midrule
        HopfE & 0.939& 0.931& \textbf{0.195}& 0.123& 0.206& 0.136& \textbf{0.951}& \textbf{0.929} \\
        HopfE + Semantics & \textbf{0.941}& \textbf{0.932}& 0.193& \textbf{0.124}& 0.253& 0.167& 0.929& 0.881\\
        \bottomrule
    \end{tabular}
    \caption{Comparison of HopfE with a setting (HopfE(w/o Hopf)) where we omit Hopf Mapping on WN18 dataset. Inducing Hopf Fibration for dimensional expressivity shows a clear empirical edge in majority relation types. Best values in bold.}
    \label{tab:tab-many2many}
        \vspace{-2mm}
\end{table*} %%%%%%%%%%%%%%%%%%%%%%%%%%%%%%%%%%%%%%%%%%%%%%%%%%%%%%%%%%%%%%%%%%%
%Models possessing geometric interpretability (RotatE, NagE) perform better than QuatE aiming for higher dimensional expressiveness. 
%%%%%%%%%%%%%%%%%%%%%%%%%%%%%%%%%%%%%%%%%%%%%%%%%%%%%%%%%%%%%%%%%%%
%%%%%%%%%%%%%%%%%%%%%%%%%%%%%%%%%%%%%%%%%%%%%%%%%%%%%%%%%%%%%%%%%%%
%%%%%%%%%%%%%%%%%%%%%%%%%%%%%%%%%%%%%%%%%%%%%%%%%%%%%%%%%%%%%%%%%%%
%%%%%%%%%%%%%%%%%%%%%%%%%%%%%%%%%%%%%%%%%%%%%%%%%%%%%%%%%%%%%%%%%%%
%%%%%%%%%%%%%%%%%%%%%%%%%%%%%%%%%%%%%%%%%%%%%%%%%%%%%%%%%%%%%%%%%%%
%%%%%%%%%%%%%%%%%%%%%%%%%%%%%%%%%%%%%%%%%%%%%%%%%%%%%%%%%%%%%%%%%%%%%%%%%%%%%%%%%%%%%%%%%%%%%%%%%%%%%%%%%%%%%%%%%%%%%%%%%%%%%%%%%%%%%%%%%%%%%%%%%%%%%%%%%%%%%%%%%%%%%%%%%%%%%%%%%%%%%%%%%%%%%%%%%%%%%%%%%%%%%%%%%%%%%%%%%%%%%%%%%%%%%%%%%%%%%%%%%%%%%%%%%%%%%%%%%%%%%%%%%%%%%%%%%%%%%%%%%%%%%%%%%%%%%%%%%%%%%%%%%%%%%%%%%%%%%%%%%%%%%%%%%%%%%%%%%%%%%%%%%%%%%%%%%%%%%%%%%%%%%%%%%%%%%%%%%%%%%%%%%%%%%%%%%%%%%%%%%%%%%%%%%%%%%%%%%%%%%%%%%%%%%%%%%%%%%%%%%%%%%%%%%%%%%%%%%%%%%%%%%%%%%%%%%%%%%%%%%%%%%%%%%%%%%%%%%%%%%%%%%%%%%%%%%%%%%%%%%%%%%%%%%%%%%%%%
% \subsection{Ablation Study}
\subsection{Ablation Studies}
\textbf{Effect of Adding Semantics:}
To understand the impact of adding textual literals, we conduct an ablation on the attribute information induced in the model. We add each entity attribute to study their individual impact. The results for the YAGO3-10 dataset are reported in Table \ref{tab:tab-ablation_semantics} because this dataset provides both textual and numeric literals. We observe that different semantic information corresponds to variations in the results. Also, we note that adding some semantic attributes causes a drop in performance similar to a decline in performance in Table \ref{tab:tab-fb_yago} when semantics are collectively added. Furthermore, encoding of numeric literals is a limitation of our approach. This is possibly due to the implementation choice we made for encoding literals using character and word embeddings. 
Hence, there are two open questions for future research: 1) how a model can choose semantic information dynamically 2) how to incorporate numeric literals in this setting for higher performance gain as observed in alternative approaches such as \cite{kristiadi2019incorporating,garcia2017kblrn}.\\
%%%%%%%%%%%%%%%%%%%%%%%%%%%%%%%%%%%%%%%%%%%%%%%%%%%%%%%%%%%%%%%%%%%%%%%%%%%%%%%%%%%%%%%%%%%%%%%%%%%%%%%%%%%%%%%%%%%%%%%%%%%%%%%%%%%%%%%%%%%%%%%%%%%%%%%%%%%%%%%%%%%%%%%%%%%%%%%%%%%%%%%%%%%%%%%%%%%%%%%%%%%%%%%%%%%%%%%%%%%%%%%%%%%%%%%%%%%%%%%%%%%%%%%%%%%%%%%%%%%%%%%%%%%%%%%%%%%%%%%%%%%%%%%%%%%%%%%%%%%%%%%%%%%%%%%%%%%%%%%%%%%%%%%%%%%%%%%%%%%%%%%%%%%%%%%%%%%%%%%%%%%%%%%%%%%%%%%%%%%%%%%%%%%%%%%%%%%%%%%%%%%%%%%%%%%%%%%%%%%%%%%%%%%%%%%%%%%%%%%%%%%%%%%%%%%%%%%%%%%%%%%%%%%%%%%%%%%%%%%%%%%%%%%%%%%%%%%%%%%%%%%%%%%%%%%%%%%%%%%%%%%%%%%%%%%%%%%%
\textbf{Effect of Hopf Fibration:} \label{sec:many-to-many-ablation}
% Mention N2N shows improved performance and the datasets which have more of such relations have good performance
Our rationale here is to understand empirical advantage of inverse Hopf Fibration. Hence, we created a variant of HopfE named "HopfE(w/o Hopf)" that only contains rotation in 3D sphere (section \ref{sec:approach-rotation}) and does not include mapping from 3D to 4D using Hopf Fibration compromising higher dimension expressiveness. For the same, we report results on individual relational patterns.
The results are in Table \ref{tab:tab-many2many}. We observe that the performance of HopfE is higher compared to the "HopfE(w/o Hopf)" setting for three types of relations. On WN18, the fraction of 1-to-N links is 0.04 percent. Hence, HopfE faces challenges in learning the corresponding mapping into higher dimensions. Overall, our experiment justifies that increasing dimensional expressivity does not compromise the overall average empirical results. 
%%%%%%%%%%%%%%%%%%%%%%%%%%%%%%%%%%%%%%%%%%%%%%%%%%%%%%%%%%%%%%%%%%%%%%%%%%%%%%%%%%%%%%%%%%%%%%%%%%%%%%%%%%%%%%%%%%%%%%%%%%%%%%%%%%%%%%%%%%%%%%%%%%%%%%%%%%%%%%%%%%%%%%%%%%%%%%%%%%%%%%%%%%%%%%%%%%%%%%%%%%%%%%%%%%%%%%%%%%%%%%%%%%%%%%%%%%%%%%%%%%%%%%%%%%%%%%%%%%%%%%%%%%%%%%%%%%%%%%%%%%%%%%%%%%%%%%%%%%%%%%%%%%%%%%%%%%%%%%%%%%%%%%%%%%%%%%%%%%%%%%%%%%%%%%%%%%%%%%%%%%%%%%%%%%%%%%%%%%%%%%%%%%%%%%%%%%%%%%%%%%%%%%%%%%%%%%%%%%%%%%%%%%%%%%%%%%%%%%%%%%%%%%%%%%%%%%%%%%%%%%%%%%%%%%%%%%%%%%%%%%%%%%%%%%%%%%%%%%%%%%%%%%%%%%%%%%%%%%%%%%%%%%%%%%%%%%%%
\textbf{Geometric interpretation:}
Having established the importance of Hopf Fibration in the previous experiment, we now study if the geometric \textit{interpretability} of rotations in $S^2$ is maintained after training with the Fibrations on WN18 dataset. We analyze the inverse relations $has\_part$ and $part\_of$ and the compositional relations $derivationally\_related\_form$ (r1 in Figure \ref{fig:geom_analysis}),  $hypernym$ (r2 in Figure \ref{fig:geom_analysis}), and $synset\_domain\_topic\_of$ (r3 in Figure \ref{fig:geom_analysis}). As illustrated in Figure \ref{fig:geom_analysis}, the histogram of the sum of the rotation angles of the inverse relations are concentrated at $0^c$ and for the compositional relations the difference between the angles is concentrated around $0^c$ as expected, maintaining geometric \textit{interpretability}.\\
\begin{figure}[t!]
  \begin{subfigure}[t]{0.22\textwidth}
    \includegraphics[width=\textwidth]{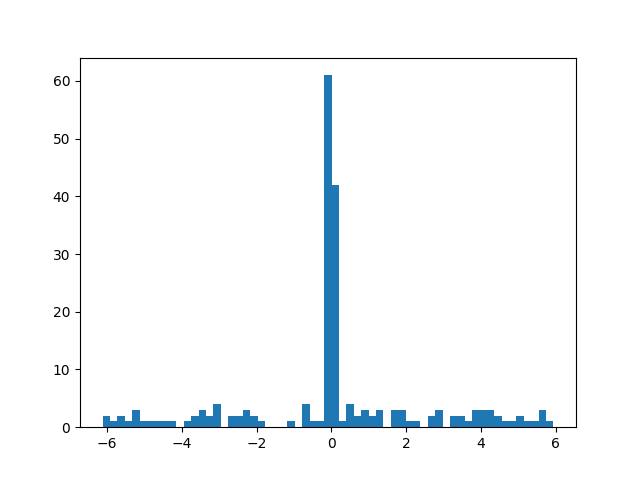}
    \caption{$\phi_{has\_part}+\phi_{part\_of}$}
    % \vspace{-2mm}
  \end{subfigure}%
  ~
  \begin{subfigure}[t]{0.22\textwidth}
    \includegraphics[width=\textwidth]{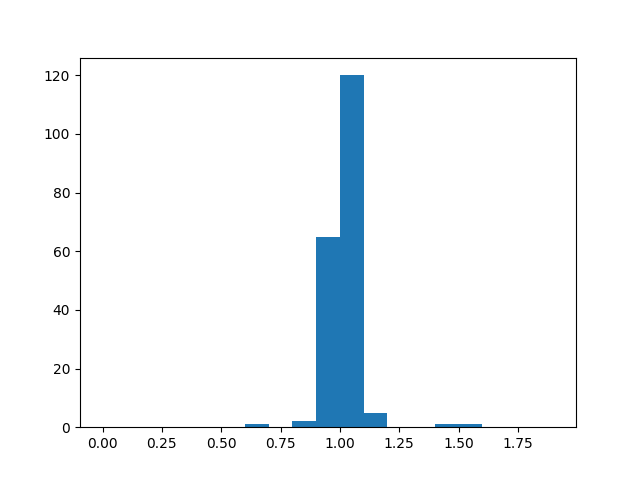}
    \caption{$\lVert r_{has\_part} \rVert * \lVert r_{has\_part} \rVert$}
    \label{fig:hist_q_inv}
    % \vspace{-2mm}
  \end{subfigure}
  
  \begin{subfigure}[t]{0.22\textwidth}
    \includegraphics[width=\textwidth]{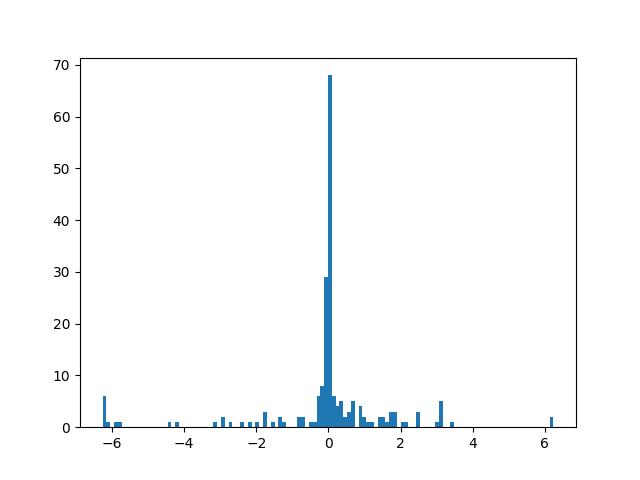}
    \caption{$\phi(r_{1} \otimes r_{2}) - \phi(r_{1})$}
    \label{fig:hist_phi_comp_101}
    % \vspace{-2mm}
  \end{subfigure}
  ~
  \begin{subfigure}[t]{0.22\textwidth}
    \includegraphics[width=\textwidth]{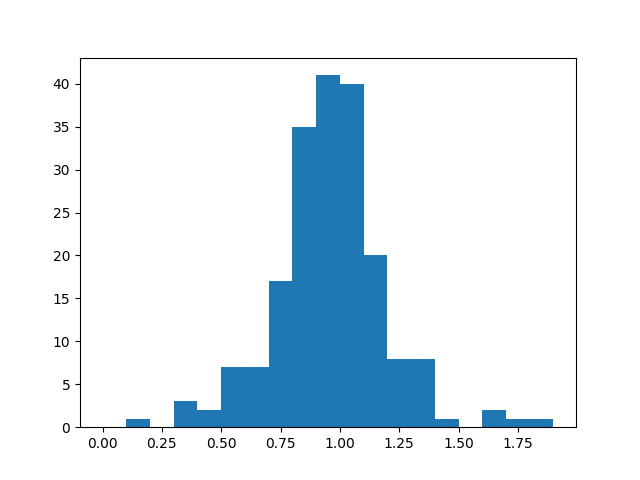}
    \caption{$\lVert r_{part\_of} \rVert$}
    \label{fig:hist_q2_inv}
    % \vspace{-2mm}
  \end{subfigure}
  
 \begin{subfigure}[t]{0.22\textwidth}
    \includegraphics[width=\textwidth]{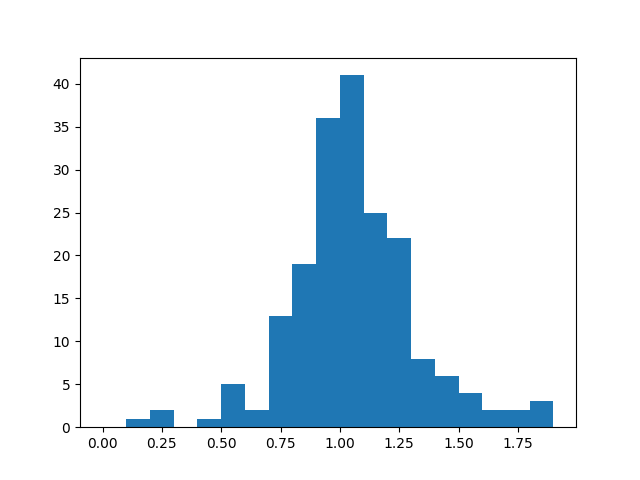}
    \caption{$\lVert r_{has\_part} \rVert$}
    \label{fig:hist_q1_inv}
    % \vspace{-2mm}
  \end{subfigure}
  ~
  \begin{subfigure}[t]{0.22\textwidth}
    \includegraphics[width=\textwidth]{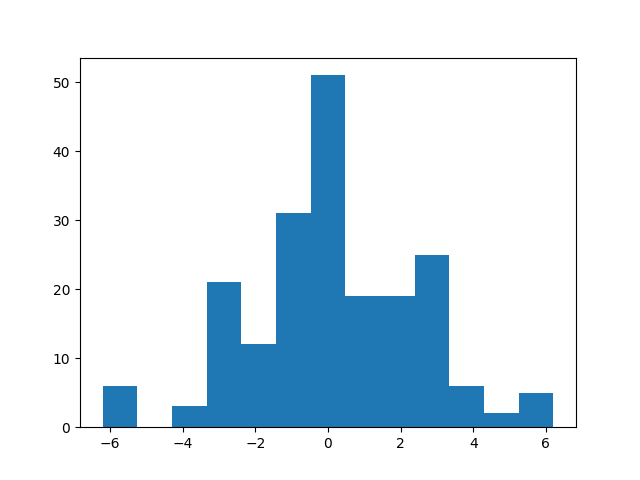}
    \caption{$\phi(r_{1} \otimes r_{2}) - \phi(r_{3})$}
    \label{fig:hist_phi_comp_105}
    % \vspace{-2mm}
  \end{subfigure}

  \caption{Analysis of the geometric properties of the inverse and compositional relations.}
  \label{fig:geom_analysis}
      \vspace{-1mm}
\end{figure}
\begin{figure}[ht!]
  \begin{subfigure}[t]{0.25\textwidth}
    \includegraphics[width=\textwidth]{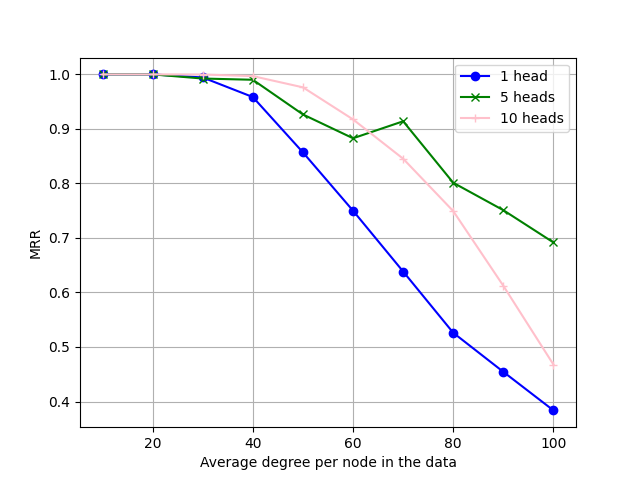}
    \caption{}
    % \vspace{-2mm}
    \label{fig:fig_synth_deg_vs_mrr}
  \end{subfigure}%
%   ~
  \begin{subfigure}[t]{0.25\textwidth}
    \includegraphics[width=\textwidth]{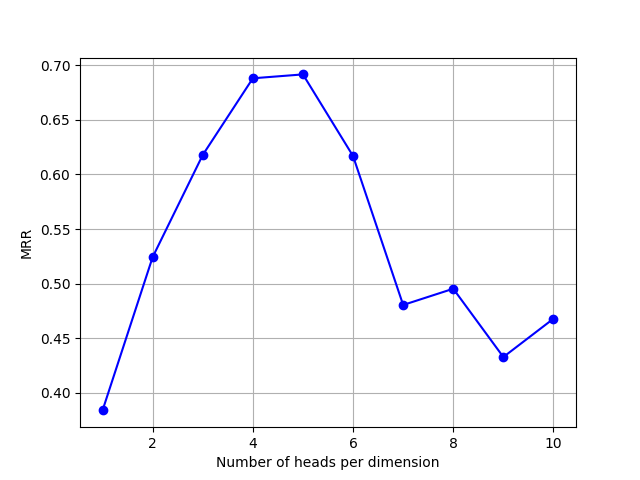}
    \caption{}
    \label{fig:fig_synth_heads_vs_mrr}
    % \vspace{-2mm}
  \end{subfigure}

  \caption{Effect of attribute heads on the performance.}
  \label{fig:fig_synth_heads/deg_vs_mrr}
     %\vspace{-1mm}
\end{figure}
\textbf{Relation wise performance:} \label{sec:relation-wise-performance}
Previous ablation study established that the geometric interpretability remains intact for HopfE. Now, we study \textit{expressiveness} in a higher dimension (in addition to Table \ref{tab:tab-wn} and Table \ref{tab:tab-fb_yago}). We perform a relation-wise performance on the WN18RR dataset to understand if HopfE can maintain \textit{expressiveness} compared to rotational methods which are in the lower dimensions. We borrow experimental settings from RotatE \cite{sun2018rotate}. As seen in Table \ref{tab:tab-rel_wise_study}, in comparison to the RotatE baseline (considering NagE skips this analysis), the HopfE model can perform better on these relations. The observed behavior implies that when we aim for geometric \textit{interpretation} using Hopf Fibration, the performance does not decrease for any relation.\\
%%%%%%%%%%%%%%%%%%%%%%%%%%%%%%%%%%%%%%%%%%%%%%%%%%%%%%%%%%%%%%%%%%%%%%%%%%%%%%%%%%%%%%%%%%%%%%%%%%%%%%%%%%%%%%%%%%%%%%%%%%%%%%%%%%%%%%%%%%%%%%%%%%%%%%%%%%%%%%%%%%%%%%%%%%%%%%%%%%%%%%%%%%%%%%%%%%%%%%%%%%%%%%%%%%%%%%%%%%%%%%%%%%%%%%%%%%%%%%%%%%%%%%%%%%%%%%%%%%%%%%%%%%%%%%%%%%%%%%%%%%%%%%%%%%%%%%%%%%%%%%%%%%%%%%%%%%%%%%%%%%%%%%%%%%%%%%%%%%%%%%%%%%%%%%%%%%%%%%%%%%%%%%%%%%%%%%%%%%%%%%%%%%%%%%%%%%%%%%%%%%%%%%%%%%%%%%%%%%%%%%%%%%%%%%%%%%%%%%%%%%%%%%%%%%%%%%%%%%%%%%%%%%%%%%%%%%%%%%%%%%%%%%%%%%%%%%%%%%%%%%%%%%%%%%%%%%%%%%%%%%%%%%%%%%%%%%%%
\textbf{Studying the effect of the number of heads:} \label{sec:heads}
In a real-world scenario, all KG entities will not have the same number of properties, and few of them may require a large number of attributes to model the KG ontology. In this experiment, we aim to understand how does adding more entity attributes for modeling dense relations impacts performance (cf., Claim \ref{claim_nclique})? Hence,
we study the effect of increasing the number of heads on the prediction performance using a synthetic dataset because this dataset allows us to have varying number of attribute of the entities. Specifically, we generate an Erdős–Rényi graph \cite{erdHos1959renyi} with an average degree varying from 10 to 100 in the increments of 10 steps. We use 1, 5, and 10 heads per entity per fiber(dimension). Also, note that since these are random graphs, we evaluate the performance on the train set as the performance on the test set is expected to be low in this setting. The dimensions(=10) and other common parameters are kept the same in all the models. The models are not all trained to convergence but for a fixed number of iterations. The Figure \ref{fig:fig_synth_deg_vs_mrr} illustrates that increasing the degree of nodes in the data in general results in a drop in the results, which is as expected. Another observation is that the setting with five heads can perform better than the one with ten heads at higher degrees. The observation possibly indicates that the latter model may require more training to fit better. A similar observation could be gleaned from Figure \ref{fig:fig_synth_heads_vs_mrr} where we study the effect of increasing the number of heads from 1 to 10 for a graph with an average degree per node as 100. We see that the metric improves from 1 to 5 heads and then begins to drop. We conclude that number of heads along with the training steps is a hyperparameter to be tuned for better results.

%%%%%%%%%%%%%%%%%%%%%%%%%%%%%%%%%%%%%%%%%%%%%%%%%%%%%%%%%%%%%%%%%%%%%%%%%%%%%%%%%%%%%%%%%%%%%%%%%%%%%%%%%%%%%%%%%%%%%%%%%%%%%%%%%%%%%%%%%%%%%%%%%%%%%%%%%%%%%%%%%%%%%%%%%%%%%%%%%%%%%%%%%%%%%%%%%%%%%%%%%%%%%%%%%%%%%%%%%%%%%%%%%%%%%%%%%%%%%%%%%%%%%%%%%%%%%%%%%%%%%%%%%%%%%%%%%%%%%%%%%%%%%%%%%%%%%%%%%%%%%%%%%%%%%%%%%%%%%%%%%%%%%%%%%%%%%%%%%%%%%%%%%%%%%%%%%%%%%%%%%%%%%%%%%%%%%%%%%%%%%%%%%%%%%%%%%%%%%%%%%%%%%%%%%%%%%%%%%%%%%%%%%%%%%%%%%%%%%%%%%%%%%%%%%%%%%%%%%%%%%%%%%%%%%%%%%%%%%%%%%%%%%%%%%%%%%%%%%%%%%%%%%%%%%%%%%%%%%%%%%%%%%%%%%%%%%%%%

\section{Conclusion} \label{sec:conclusion}
Our central research question was to study if a KGE approach can achieve geometric \textit{interpretability} along with the \textit{expressiveness} of four-dimensional space. The proposed theoretical foundations and associated empirical results of HopfE show that it is possible to maintain the geometric \textit{interpretability} and simultaneously increase the model’s \textit{expressiveness} in four-dimensional space.
Furthermore, HopfE enables encoding of the entity attributes to a fibration in a higher dimension for enriching the embeddings with the semantic information prevalent in real-world KGs. Our ablation studies systematically support several choices we made in the paper's scope, such as the use of Hopf Fibration, the effect of a number of heads, and the individual impact of entity semantic properties.  Besides providing empirical evidences complementing the theoretical foundations (Sections \ref{sec:theoretical} and \ref{sec:method}), our results open up important research questions for future research: 
1) In four dimensions space, how to utilize the semantic attributes (an essential aspect of KGE \cite{paulheim2018make}) more optimally to positively impacting the link prediction performance and induce semantic class hierarchies similar to \cite{zhang2020learning}. 2) How to utilize logic on the attributes (or description logic) for fine-grained reasoning and logical interpretability of a KGE model in 4D? We point readers to the open research directions as viable next steps.

\bibliographystyle{ACM-Reference-Format}
\bibliography{references}

\end{document}